\documentclass{amsart}

\usepackage{amssymb, amscd,stmaryrd,setspace,hyperref,color}

\input xy
\xyoption{all} \xyoption{poly} \xyoption{knot}\xyoption{curve}

\newcommand{\comment}[1]{}

\def\tn{\textnormal}

\def\mc{\mathcal}

\def\ZZ{{\mathbb Z}}

\def\RR{{\mathbb R}}

\def\bD{{\bf \Delta}}
\def\Str{{\bf Str}}

\def\Hom{\tn{Hom}}

\def\to{\rightarrow}
\def\from{\leftarrow}
\def\cross{\times}
\def\taking{\colon}
\def\inj{\hookrightarrow}

\def\too{\longrightarrow}
\def\fromm{\longleftarrow}
\def\down{\downarrow}

\def\ss{\subset}

\def\iso{\cong}
\def\|{{\;|\;}}
\def\m1{{-1}}
\def\op{^\tn{op}}

\def\ol{\overline}

\def\setto{\colon\hspace{-.25cm}=}

\def\ullimit{\ar@{}[rd]|(.3)*+{\lrcorner}}
\def\urlimit{\ar@{}[ld]|(.3)*+{\llcorner}}
\def\lllimit{\ar@{}[ru]|(.3)*+{\urcorner}}
\def\lrlimit{\ar@{}[lu]|(.3)*+{\ulcorner}}
\def\ulhlimit{\ar@{}[rd]|(.3)*+{\diamond}}
\def\urhlimit{\ar@{}[ld]|(.3)*+{\diamond}}
\def\llhlimit{\ar@{}[ru]|(.3)*+{\diamond}}
\def\lrhlimit{\ar@{}[lu]|(.3)*+{\diamond}}
\newcommand{\clabel}[1]{\ar@{}[rd]|(.5)*+{#1}}

\newcommand{\arr}[1]{\ar@<.5ex>[#1]\ar@<-.5ex>[#1]}
\newcommand{\arrr}[1]{\ar@<.7ex>[#1]\ar@<0ex>[#1]\ar@<-.7ex>[#1]}
\newcommand{\arrrr}[1]{\ar@<.9ex>[#1]\ar@<.3ex>[#1]\ar@<-.3ex>[#1]\ar@<-.9ex>[#1]}
\newcommand{\arrrrr}[1]{\ar@<1ex>[#1]\ar@<.5ex>[#1]\ar[#1]\ar@<-.5ex>[#1]\ar@<-1ex>[#1]}

\newcommand{\To}[1]{\xrightarrow{#1}}

\newcommand{\From}[1]{\xleftarrow{#1}}

\newcommand{\Adjoint}[4]{\xymatrix@1{#2 \ar@<.5ex>[r]^-{#1} & #3 \ar@<.5ex>[l]^-{#4}}}
\newcommand{\adjoint}[4]{\xymatrix@1{#1\colon #2\ar@<.5ex>[r]& #3\;:#4 \ar@<.5ex>[l]}}

\def\id{\tn{id}}
\def\Top{{\bf Top}}
\def\Cat{{\bf Cat}}
\def\Sets{{\bf Sets}}

\def\Pre{{\bf Pre}}
\def\Shv{{\bf Shv}}

\def\colim{\mathop{\tn{colim}}}

\def\mcA{\mc{A}}
\def\mcB{\mc{B}}
\def\mcC{\mc{C}}
\def\mcD{\mc{D}}

\def\mcK{\mc{K}}
\def\mcL{\mc{L}}

\def\mcR{\mc{R}}
\def\mcS{\mc{S}}

\def\mcU{\mc{U}}

\def\mcX{\mc{X}}
\def\mcY{\mc{Y}}

\newtheorem{theorem}{Theorem}[subsection]
\newtheorem{lemma}[theorem]{Lemma}
\newtheorem{proposition}[theorem]{Proposition}
\newtheorem{corollary}[theorem]{Corollary}

\theoremstyle{remark}
\newtheorem{remark}[theorem]{Remark}
\newtheorem{example}[theorem]{Example}

\newtheorem{question}[theorem]{Question}
\newtheorem{guess}[theorem]{Guess}

\newtheorem{construction}[theorem]{Construction}

\theoremstyle{definition}
\newtheorem{definition}[theorem]{Definition}

\def\DT{{\bf DT}}

\def\DB{\GD}
\def\Sch{{\bf Sch}}

\def\Strings{{\bf Strings}}
\def\ND{{\bf ND}}
\def\Tables{{\bf Tables}}
\def\'{\tn{'}}

\def\Rel{{\bf Rel}}
\def\mcRel{{\bf \mcR el}}
\def\Cech{$\check{\tn{C}}$ech }
\def\C{\check{\tn{C}}}

\def\singleton{\{*\}}
\def\Sub{{\bf Sub}}
\def\card{\tn{card}}
\def\Data{{\bf DB}}
\def\DB{{\bf DB}}
\def\im{\tn{im}}
\def\'{\tn{'}}

\usepackage{graphicx}

\begin{document}

\author{David I. Spivak}

\thanks{This project was supported in part by the Office of Naval Research.}

\title{Simplicial Databases}

\begin{abstract}

In this paper, we define a category $\Data$, called the category of simplicial databases, whose objects are databases and whose morphisms are data-preserving maps.  Along the way we give a precise formulation of the category of relational databases, and prove that it is a full subcategory of $\Data$.  We also prove that limits and colimits always exist in $\Data$ and that they correspond to queries such as select, join, union, etc.

One feature of our construction is that the schema of a simplicial database has a natural geometric structure: an underlying simplicial set.  The geometry of a schema is a way of keeping track of relationships between distinct tables, and can be thought of as a system of foreign keys.  The shape of a schema is generally intuitive (e.g. the schema for round-trip flights is a circle consisting of an edge from $A$ to $B$ and an edge from $B$ to $A$), and as such, may be useful for analyzing data.  

We give several applications of our approach, as well as possible advantages it has over the relational model.  We also indicate some directions for further research.

\end{abstract}

\maketitle

\setcounter{tocdepth}{1}

\tableofcontents

\section{Introduction}\label{sec:intro}

The theory of relational databases is generally formulated within mathematical logic.  We provide a more modern and more flexible approach using methods from category theory and algebraic topology.  Category theory is useful both as a language and as a tool, and has been successfully applied to many areas of computer science.  Using an inefficient language can hamper ones ability to implement, work with, and reason about a subject.  This can be seen as one reason that SQL implements tables, rather than relational databases in their pure form: perhaps mathematical logic is not a sufficiently flexible language for discussing databases as they are used in practice.

One reason that relational databases have been so successful is that their definition can be phrased within a precise mathematical language.  The definition we provide in this paper is just as precise, if not more so (see the discussion at the beginning of Section \ref{sec:schemas and databases}).  However, we go beyond simply defining the {\em objects} of study (databases), but instead continue on to define {\em morphisms} between databases.  With these definitions, we have a category of databases.

There are many categories whose objects are databases (the difference being in their morphisms); what makes one definition better than another?  First, a good definition should make sense -- the morphisms should somehow preserve the structure of the databases.  Second, applying common categorical constructions (colimits, limits, etc.) to the category of databases should result in common database constructions, such as unions, joins, etc.  Third, the categorical approach should make reasoning about databases, such as that needed for maintaining and restructuring databases, easier.  

Our formulation accomplishes these three goals (see Remark \ref{rem:data integrity}, and Sections \ref{sec:constructions for databases} and \ref{sec:applications}, respectively).  As an added bonus, the schemas for our databases have geometric structure (more precisely, the structure of a simplicial set).  In other words, the schema is given as a geometric object which one should think of as a kind of Entity-Relationship diagram for the schema.  This approach may lead to improvements in query optimization because one can adjust the ``shape" of the schema to fit with the purposes of the queries to be taken.  The ability to visualize data should also prove useful, because these visualizations seem to ``make sense" in practice.  Examples of this phenomenon are given in \ref{ex:flights} and \ref{ex:sex}, where we respectively discuss round trip flights and a sociological experiment involving 4-cycles in high school partnerships.

The data on a given schema is given by a sheaf of sets on that schema.  Sheaves are ubiquitous in modern mathematics because they generalize sets and functions and because they have good formal properties.  Classical operations on sheaves (such as direct images) allow one to transport data from one schema to another in a functorial way.  One of the main purposes of this paper is to provide a good language for discussing databases mathematically, and the consideration of data as a sheaf on a given schema helps to accomplish that goal.

Other researchers have formulated databases in terms of category theory (for example, see \cite{RW},\cite{JRW},\cite{PS},\cite{Ber},\cite{DK},\cite{Dis},\cite{GB}).  Of note is work by Cadish and Diskin, and work by Rosebrugh and Wood.  There are many differences between previous viewpoints and our own.   Most notably, our work uses simplicial methods to give a geometric structure to the schemas of databases and uses sheaves over these spaces to model the data itself.  Both of these approaches appear to be new. 

We assume throughout this paper that the reader has a basic knowledge of category theory which includes knowing the definition of category, functor, limit, and colimit, as well as basic facts such as Yoneda's lemma.  Good references for this material include \cite{Mac},\cite{BW}, and \cite{Bor1}.  We do not assume that the reader has a prior knowledge of sheaves or of simplicial sets.

We begin by defining the category of tables, in Section \ref{sec:tables}.  In Section \ref{sec:constructions for tables}, we prove that the category of tables is closed under limits and certain colimits, and that these constructions correspond to joins and unions.  We also prove that projections and deletions are easily defined under our formulation.  In Section \ref{sec:schemas and databases}, we first give a brief description of simplicial sets.  We then proceed to define the category of simplicial databases.  In Section \ref{sec:constructions for databases}, we prove that the category of simplicial databases is closed under all limits and colimits and prove that they again correspond to joins and unions.  Finally in Section \ref{sec:applications}, we discuss some applications of our model and directions for future research.

\subsection{Acknowledgments}

I would like to thank Paea LePendu for explaining relational databases to me, for suggesting that databases should be categorified, and for his advice and encouragement throughout the process.  I would also like to thank Chris Wilson for several useful conversations.

\section{The category of Tables}\label{sec:tables}

It is no accident that SQL uses tables instead of relations: Tables are inherently more useful, yet just as easy to implement.  They are disliked by the purists of relational database theory not because they are bad, but because they do not fit in with that theory.  In this section we provide a categorical structure to the set of tables, thus firmly grounding it in rigorous mathematics.  

\subsection{Data types}

In order to define schemas, records, and tables of a given type, we need to define what we mean by ``type." 

\begin{definition}

A {\em type specification} is simply a function between sets $\pi\taking U\to \DT$.  The set $\DT$ is called the set of {\em data types} for $\pi$, and the set $U$ is called the {\em domain bundle} for $\pi$.  Given any element $T\in\DT$, the preimage $\pi^\m1(T)\ss U$ is called the {\em domain of $T$}, and an element $x\in\pi^\m1(T)$ is called an {\em object of type $T$}.

\end{definition}

\begin{example}\label{ex:type specification}

Let $U$ denote the disjoint union $U\setto (\ZZ\amalg\RR\amalg\Strings)$ and let $\DT$ denote the three element set $\{`\ZZ\', `\RR\', `\Strings\'\}$.  Let $\pi\taking U\to\DT$ denote the obvious function, which send all of $\ZZ$ to the element $`\ZZ\'$, all of $\RR$ to $`\RR\'$, and all of $\Strings$ to $`\Strings\'$. 
The preimage $\pi^\m1(`\Strings\')\ss U$, which we have called the domain of the type $`\Strings\'$, is indeed the set of strings.

As another example, the mod 2 function $\pi\taking\ZZ\to\{\tn{`even'},\tn{`odd'}\}$ is a type specification in which the objects of type `even' are the even integers.

\end{example}

\subsection{Schemas}

We quickly recall the definition of fiber product (for sets).

\begin{definition}

Let $A, B,$ and $C$ be sets, and suppose $f\taking A\to B$ and $g\taking C\to B$ are functions with the same codomain.  The {\em fiber product of $A$ and $C$ over $B$}, denoted $A\cross_BC$, is the set $$A\cross_BC\setto \{(a,c)\in A\cross C | f(a)=g(c)\in B\}.$$  The fiber product moreover comes equipped with obvious projection maps making the diagram $$\xymatrix{A\cross_BC\ar[r]^-{f'}\ar[d]_{g'}\ullimit&C\ar[d]^g\\A\ar[r]_f&B}$$ commute.  The corner symbol $\lrcorner$ serves to remind the reader that the object in the upper left is a fiber product.   We sometimes call $g'\taking A\cross_BC\to A$ the {\em pullback of $g$ along $f$}; similarly $f'$ is the pullback of $f$ along $g$.

\end{definition}

\begin{remark}

The fiber product of the diagram $A\To{f}B\From{g}C$ above should probably be denoted $f\cross_Bg$ instead of $A\cross_BC$, since it depends on the maps $f$ and $g$, not just their domains.  However, this is not often done, and in this paper the maps will be clear from context.

\end{remark}

\begin{definition}\label{def:simple schema}

Let $\pi\taking U\to\DT$ denote a type specification.  A {\em simple schema of type $\pi$} consists of a pair $(C,\sigma)$, where $C$ is a finite (totally) ordered set and $\sigma\taking C\to\DT$ is a function.   We sometimes denote the simple schema $(C,\sigma)$ by $\sigma$.  We refer to $C$ as the {\em column set} or {\em set of attributes} for $\sigma$ and $\pi$ as the {\em type specification} for $\sigma$.

Let $U_\sigma\setto \sigma^\m1(U)$ denote the fiber product $U\cross_\DT C$.  We call the pullback $\pi_\sigma\taking U_\sigma\to C$, i.e. the left hand map in the diagram $$\xymatrix{U_\sigma\ar[r]\ar[d]_{\pi_\sigma}\ullimit&U\ar[d]^\pi\\C\ar[r]_\sigma&\DT,}$$ the {\em domain bundle on $C$} induced by $\sigma$.

\end{definition}

\begin{remark}\label{rem:order}

We do not worry much about the ordering on $C$, as evidenced by the fact that we do not record it in the notation $(C,\sigma)$ for the simple schema.  In fact the ordering requirement can be dropped from the definition if one so chooses.  

The reason we include it is first because the columns of a displayed table naturally come with an order (left to right), and second because it results in a more commonly used mathematical object down the road in Section \ref{sec:schemas and databases}.  See Remark \ref{rem:Grandis}.

\end{remark}

\begin{example}\label{ex:schema}

Let $\pi\taking U\to\DT$ denote the type specification of Example \ref{ex:type specification}.  Let $C=(\tn{`First Name', `Last Name',`Age'})$, and define $\sigma\taking C\to\DT$ by \begin{align*}\sigma(\tn{`First Name'})&=`\Strings\'\\\sigma(\tn{`Last Name'})&=`\Strings\'\\\sigma(\tn{`Age'})&=`\ZZ'\end{align*}   We see that $C$ is a set of attributes for the simple schema $\sigma$.  We call $C$ the column set because, once we arrange data in terms of tables, the columns of these tables will each be headed by an element of $C$.

One can check that the domain bundle $U_\sigma\to C$ induced by $\sigma$ is the obvious function $$(\Strings\amalg\Strings\amalg\ZZ)\too C.$$  Thus an object of type `First Name' is a string in this example.

\end{example}

\begin{definition}\label{def:category of schema}

Let $\pi\taking U\to\DT$ denote a type specification.  A {\em morphism of simple schemas (of type $\pi$)}, written $f\taking (C,\sigma)\to (C',\sigma')$, is an order-preserving function $f\taking C\to C'$ such that the triangle $$\xymatrix@=16pt{C\ar[rr]^f\ar[dr]_\sigma&&C'\ar[dl]^{\sigma'}\\&\DT}$$ commutes.

The {\em category of simple schemas on $\pi$}, denoted $\mcS^\pi$ is the category whose objects are simple schemas and whose morphisms are morphisms thereof.  

\end{definition}

\begin{remark}\label{rem:bD down DT}

Let $\bD$ denote the category of finite ordered sets.  Let $(\bD\down\DT)$ denote the category for which an object is a finite ordered set with a map to $\DT$ and for which a morphism is an order-preserving function, over $\DT$.  One can easily see that the category $\mcS^\pi$ is isomorphic to $(\bD\down\DT)$, regardless of $\pi$.  However, we should think of $\pi$ as part of the data for a simple schema.  

Note that the symbol $\bD$ typically refers to the category of {\em non-empty} finite ordered sets; one typically denotes the category of all finite ordered sets as $\bD_+$.  For typographical reasons, we do not follow the standard convention in this paper.  

\end{remark}

\subsection{Records and Tables}

\begin{definition}\label{def:records}

Let $(C,\sigma)$ be a simple schema.  A {\em record on $(C,\sigma)$} is a function $r\taking C\to U_\sigma$ such that $\pi_\sigma\circ r=\id_C$, i.e. a section of the domain bundle for $\sigma$.  We denote the set of records on $\sigma$ by $\Gamma^\pi(\sigma)$, or simply by $\Gamma(\sigma)$ if $\pi$ is understood.

\end{definition}

In other words, a record must produce, for each attribute $c\in C$, an object of type $\sigma(c)\in\DT$.  

\begin{example}\label{ex:record}

Let $\pi$ and $(C,\sigma)$ be as in Example \ref{ex:schema}.  A record on that simple schema is a section $r$ as depicted in the diagram $$\xymatrix{\Strings\amalg\Strings\amalg\ZZ\ar[d]^{\pi_\sigma}\\\{\tn{`First Name', `Last Name',`BYear'}\}.\ar@/^1pc/[u]^r}$$  That is, a record is a way to designate a first name and a last name (in $\Strings$) and an age (in $\ZZ$).  For example (Barack; Obama; 1961) denotes a record on this simple schema; that is, it defines a section of $\pi_\sigma$.

The set $\Gamma(\sigma)$ of records on $(C,\sigma)$ is simply the set of all possible such sections.  In this example $\Gamma(\sigma)=\Strings\cross\Strings\cross\ZZ$.

\end{example}

\begin{definition}\label{def:tables}

Let $\pi\taking U\to\DT$ be a type specification.  A {\em table of type $\pi$} consists of a sequence $(K,C,\sigma,\tau)$, where $K$ is a set, $(C,\sigma)$ is a simple schema of type $\pi$,  and $\tau\taking K\to\Gamma(\sigma)$ is a function.  We sometimes denote the table $(K,C,\sigma,\tau)$ simply by $\tau$.  The set $K$ is called the {\em set of keys of $\tau$}, and $(C,\sigma)$ is called the {\em simple schema of $\tau$}.

\end{definition}

\begin{remark}

Given a table $(K,C,\sigma,\tau)$, those familiar with SQL should think of the set $K$ of keys as the set of row identifiers for a table.  These row ids are always unique identifiers and serve as an internal key system for the table; they are generally not considered as part of the data.  

\end{remark}

\begin{remark}

We do not require our tables to have finitely many rows.  One could easily enforce such a restriction if desired, and follow the rest of the paper with that restriction in mind.  The resulting category would be a full subcategory of the one we present in Definition \ref{def:category of tables}, it would still be closed under finite limits (etc.), and queries would be taken in precisely the same way as they are here.

\end{remark}

\begin{example}\label{ex:table}

Given a simple schema $(C,\sigma)$, a table on it is simply a collection of records indexed by a set $K$.  The records need not be distinct because the set $K$ keeps track of the distinctions.  Continuing with $\pi$ and $(C,\sigma)$ as in Example \ref{ex:record}, we could have $K=\{1,2,`foo'\}$ and let $\tau\taking K\to\Gamma(\sigma)$ be the assignment \begin{align*} 1&\mapsto \tn{(Barack; Obama; 1961)}\\2&\mapsto \tn{(Michelle; Obama; 1964)}\\ `foo'&\mapsto\tn{(Barack; Obama; 1961)}\end{align*}  This table can be written in more standard form as: $$\begin{tabular}{|l||l|l|l|}\hline K&`First Name'&`Last Name'&`BYear'\\\hline\hline 1&Barack&Obama&1961\\\hline 2&Michelle&Obama&1964\\\hline `foo'&Barack&Obama&1961\\\hline\end{tabular}$$  We indicate with the double vertical line the fact that this table corresponds to a function whose domain is $K$.

\end{example}

\begin{lemma}\label{induced morphisms}

Let $\pi\taking U\to\DT$ denote a type specification, let $(C_1,\sigma_1)$ and $(C_2,\sigma_2)$ denote simple schemas on $\pi$, and let $f\taking (C_2,\sigma_2)\to (C_1,\sigma_1)$ denote a morphism of simple schemas.  There is an induced map on record sets $f^*\taking \Gamma(\sigma_1)\to\Gamma(\sigma_2)$.

\end{lemma}

\begin{proof}

Consider the diagram $$\xymatrix{U_{\sigma_2}\ar[r]\ar[d]_{\pi_2}\ullimit &U_{\sigma_1}\ar[r]\ar[d]_{\pi_1}\ullimit&U\ar[d]_\pi\\C_2\ar[r]^f\ar@/_1pc/[rr]_{\sigma_2}&C_1\ar[r]^{\sigma_1}&\DT.}$$  Note that the left hand square is a fiber product square.  This follows by applying basic category theory (specifically the ``pasting lemma" for fiber products; see \cite{Mac}) to the fact that the right hand square and the big rectangle are fiber product squares.  We must show that a section $r_1\taking C_1\to U_{\sigma_1}$ of $\pi_1$ induces a section $r_2\taking C_2\to U_{\sigma_2}$ of $\pi_2$, because this assignment will constitute $f^*\taking\Gamma(\sigma_1)\to\Gamma(\sigma_2)$.

Suppose given $r_1$ with $\pi_1\circ r_1=\id_{C_1}$.   We have a map $r_1\circ f\taking C_2\to U_{\sigma_1}$ and a map $\id_{C_2}\taking C_2\to C_2$ such that $f\circ\id_{C_2}=f=\pi_1\circ(r_1\circ f)$.  By the universal property, these two maps define a map $r_2\taking C_2\to U_{\sigma_2}$ such that, in particular $\pi_2\circ r_2=\id_{C_2}$.  This is the desired section of $\pi_2$.

\end{proof}

Given a morphism $f\taking\sigma_2\to\sigma_1$ of simple schemas, the function $f^*\taking\Gamma(\sigma_1)\to\Gamma(\sigma_2)$ defined in the above lemma is said to be {\em induced} by $f$.

\begin{definition}\label{def:morphism of tables}

Let $\pi\taking U\to\DT$ be a type specification, and let $(K_1,C_1,\sigma_1,\tau_1)$ and $(K_2,C_2,\sigma_2,\tau_2)$ denote tables.  A {\em morphism of tables} $\varphi\taking\tau_1\to\tau_2$ consists of a pair $(g,f)$, where $g\taking K_1\to K_2$ is a function and $f\taking (C_2,\sigma_2)\to(C_1,\sigma_1)$ is a morphism of simple schema such that the diagram of sets \begin{align}\label{dia:morphism of tables}\xymatrix{K_1\ar[r]^{\tau_1}\ar[d]_g&\Gamma(\sigma_1)\ar[d]^{f^*}\\K_2\ar[r]_{\tau_2}&\Gamma(\sigma_2)}\end{align} commutes, where $f^*\taking\Gamma(\sigma_1)\to\Gamma(\sigma_2)$ is the function induced by $f$.

\end{definition}

\begin{example}\label{ex:morphism of tables}

Let us continue with Example \ref{ex:table}, except for a slight renaming of objects: $C_1\setto C, \sigma_1\setto \sigma, K_1\setto K,$ and $\tau_1\setto \tau$.  Let $C_2=\{\tn{`First', `Last'}\}$ and let $\sigma_2$ send both elements to the data type $\Strings\in\DT$; thus $\Gamma(\sigma_2)=\Strings\cross\Strings$.  

Let $K_2=\{5,6,`bar'\}$ and $\tau_2$ be the assignment \begin{align*}5&\mapsto \tn{(Barack; Obama)}\\6 &\mapsto\tn{(Michelle; Obama)}\\`bar' &\mapsto\tn{(George; Bush)}.\end{align*}  

A morphism of tables $\varphi\taking\tau_1\to\tau_2$ should consist of a map $g\taking K_1\to K_2$ and a map $f^*\taking\Gamma(C_1)\to\Gamma(C_2)$.   We have an obvious map of simple schema $f\taking C_2\to C_1$, namely $\tn{`First'}\mapsto\tn{`First name'}$ and $\tn{`Last'}\mapsto\tn{`Last name'}$.  Then $f^*\taking\Gamma(\sigma_1)\to\Gamma(\sigma_2)$ is just the projection $\Strings\cross\Strings\cross\ZZ\to\Strings\cross\Strings$.

Now, to define a morphism of tables $\varphi\taking\tau_1\to\tau_2$, our choice of $g$ must send both of the records $(\tn{Barack; Obama; 1961})$ in $\tau_1$ to the record $(\tn{Barack; Obama})$ and send the record $(\tn{Michelle; Obama; 1964})$ to the record $(\tn{Michelle; Obama})$.   There is a unique such morphism $\phi$ in this case. 

For a variety of reasons, there {\em does not} exist a morphism of tables $\tau_2\to\tau_1$.

\end{example}

\begin{remark}\label{rem:table integrity}

The morphism of tables in Example \ref{ex:morphism of tables} has a common form.  As in the example, a morphism of tables often is composed of a projection (in the columns) together with an inclusion (in the rows).  The requirement that the square (\ref{dia:morphism of tables}) in Definition \ref{def:morphism of tables} commutes is simply the requirement that morphisms preserve the integrity of the data.

\end{remark}

\subsection{The category of tables}

We have now defined tables and morphisms between tables.  Given morphisms depicted $$\xymatrix{K_1\ar[d]\ar[r]^{\tau_1}&\Gamma(\sigma_1)\ar[d]\\K_2\ar[d]\ar[r]^{\tau_2}&\Gamma(\sigma_2)\ar[d]\\K_3\ar[r]^{\tau_3}&\Gamma(\sigma_3)}$$ it is easy to see how composition is defined.  It is also easy to understand the identity morphism on a table $\tau\taking K\to\Gamma(C)$.  Thus we have a category.

\begin{definition}\label{def:category of tables}

Let $\pi\taking U\to\DT$ denote a type specification.  The category whose objects are tables $K\to\Gamma(\sigma)$ and whose morphisms are commutative squares as in Definition \ref{def:morphism of tables} is called {\em the category of tables on $\pi$} and is denoted $\Tables^\pi$, or simply $\Tables$, if $\pi$ is understood.

\end{definition}

\begin{example}

Suppose $\pi\taking U\to\DT$ is as in Example \ref{ex:schema}.  Suppose that $C=\{c_1,c_2\}$ and $C'=\{c_1'\}$, and that $\sigma\taking C\to\DT$ and $\sigma'\taking C'\to\DT$ are the unique maps such that $\Gamma(\sigma)=\ZZ\cross\ZZ$ and $\Gamma(\sigma')=\ZZ$.  Let $K$ and $K'$ be any two sets and $\tau\taking K\to\Gamma(\sigma)$ and $\tau'\taking K'\to\Gamma(\sigma')$ be any two tables.  

For a morphism $\tau_1\to\tau_2$ in the category of tables, we are allowed any kind of function between  key sets $K\to K'$, but the only permitted maps $\ZZ\cross\ZZ\too\ZZ$ are the two projections, because they are the only maps which are induced by morphisms of simple schema.

\end{example}

\begin{definition}

Let $\pi\taking U\to\DT$ denote a type specification and let $\sigma\taking C\to\DT$ denote a simple schema.  The {\em category of tables on $\sigma$ of type $\pi$}, denoted $\Tables^\pi_\sigma$ is the category whose objects are tables $\tau\taking K\to\Gamma(\sigma)$ and whose morphisms are triangles $$\xymatrix@=9pt{K_1\ar[drrr]^{\tau_1}\ar[dd]_g\\&&&\Gamma(\sigma)\\K_2\ar[urrr]_{\tau_2}}$$ denoted by $g\taking\tau_1\to\tau_2$.  

\end{definition}

\subsection{Relational tables}

The most common formulation of databases used today is the relational model, invented by E.F. Codd (see \cite{Cod}).  It is based on the theory of mathematical logic, and more specifically on relations.  One can find a modern treatment of the subject in \cite{Dat}.  We define a relation in Definition \ref{def:relational table} as a type of table, where the map $\tau\taking K\to\Gamma(\sigma)$ is required to be an injection.

\begin{definition}\label{def:relational table}

Let $\pi\taking U\to\DT$ denote a type specification, and let $\sigma\taking C\to\DT$ denote a simple schema on $\pi$.  A {\em relation on $\sigma$} is a table $\tau\taking K\to\Gamma(\sigma)$ for which $\tau$ is an injective function.  

A morphism of relations is a morphism of tables, for which the source and target tables are relations.  That is, the {\em category of relations}, denoted $\Rel^\pi$ is the full subcategory of $\Tables^\pi$ spanned by the relations.  Similarly, given a simple schema $\sigma$, the {\em category of relations on $\sigma$} is the full subcategory of $\Tables^\pi_\sigma$ spanned by the relations.  As usual the superscript $\pi$ can be dropped if it is understood.

There is a functor $\Rel\to\Tables$ and a functor $\Rel_\sigma\to\Tables_\sigma$, both of which are simply inclusions of full subcategories.

\end{definition}

\section{Constructions and formal properties of Tables}\label{sec:constructions for tables}

\setcounter{subsection}{1}

Our definition for the category of tables (Definition \ref{def:category of tables}) is sensible because objects are tables and morphisms are data-preserving maps.  In this section we show that category-theoretic operations on tables correspond to operations on databases, such as joins and other queries.  Fix a type specification $\pi\taking U\to\DT$ for the remainder of the section.  We will drop $\pi$ as a superscript in this section; for example the category $\mcS^\pi$ of simple schema on $\pi$ will be denoted simply by $\mcS$.

We sometimes refer to the underlying keys or underlying simple schema of a table, so we record these trivial constructions in a remark.

\begin{remark}\label{rem:forget}

There is a forgetful functor $\Tables\to\Sets$ given by sending a table $\tau\taking K\to\Gamma(\sigma)$ to the key set $K$ and a morphism of tables to the underlying map of keys.  There is another forgetful functor $\Tables\to\mcS\op$ which sends the table  $\tau$ to its simple schema $\sigma$ and a morphism $\varphi=(g,f)$ of tables to the underlying morphism of simple schema $f$. 

\end{remark}

\begin{lemma}\label{final object}

There exists a final object and an initial object in $\Tables$.    

\end{lemma}

\begin{proof}

One checks immediately that if we take $K$ to be a terminal object in $\Sets$ (i.e. any set $K$ with cardinality 1) and $\sigma$ to be the inital object $\emptyset\to\DT$ in $\mcS$, then there is exactly one table with these as its underlying keys and simple schema, and this table is the terminal object in $\Tables$.

One also checks immediately that if we take $K=\emptyset$ to be the initial object in $\Sets$ and $\sigma=\id_{\DT}\taking\DT\to\DT$ to be the final object in $\mcS$, then there is exactly one table with these as its underlying keys and simple schema, and this table is the initial object in $\Tables$.

\end{proof}

Certain colimits exist in $\Tables$; namely colimits of diagrams that are constant in the underlying simple schema.

\begin{construction}

Let $\tau_1\taking K_1\to\Gamma(\sigma)$ and $\tau_2\taking K_2\to\Gamma(\sigma)$ be two tables with the same simple schema.  By taking the disjoint union of $K_1$ and $K_2$ we get a new table $\tau\taking K_1\amalg K_2\to\Gamma(\sigma)$.  This query is called UNION ALL in SQL.  

We can also take the (non-disjoint) union of these two tables, if we know how they overlap.  That is, if there is some set $K$ with maps $g_1\taking K\to K_1$ and $g_2\taking K\to K_2$ in such a way that $\tau_1\circ g_1=\tau_2\circ g_2$, then we can obtain a new table $\tau\taking K_1\amalg_KK_2\to\Gamma(\sigma)$.  This query is called UNION in SQL.

\end{construction}

We will see that limits in the category of tables correspond to generalized joins.

\begin{proposition}\label{finite limits exist}

All finite limits exist in $\Tables$.

\end{proposition}

\begin{proof}

It suffices (see, for example, \cite[p. 30]{MM}) to show that $\Tables$ has a terminal object and is closed under taking fiber products; the first of these facts was shown in Lemma \ref{final object}.  For the second, suppose we have a diagram \begin{align}\label{eqn:limit}\xymatrix{K_1\ar[d]\ar[r]^{\tau_1}&\Gamma(\sigma_1)\ar[d]^{f_1^*}\\K\ar[r]^\tau&\Gamma(\sigma)\\K_2\ar[r]^{\tau_2}\ar[u]&\Gamma(\sigma_2)\ar[u]_{f_2^*}}\end{align} in $\Tables$, where $\sigma\taking C\to\DT$ and $\sigma_i\taking C_i\to\DT$ for $i=1,2$ are simple schemas.  As indicated, the maps $\Gamma(\sigma_i)\to\Gamma(\sigma)$ are induced by morphisms of simple schema $f_i\taking\sigma\to\sigma_i$, for $i=1,2$.

Consider the simple schema $$(\sigma_1\amalg_\sigma\sigma_2)\taking C_1\amalg_CC_2\too\DT$$ induced by taking the colimit of the column sets.  We would like to show that the natural function \begin{eqnarray}\label{eqn:limit schema}\Gamma(\sigma_1\amalg_\sigma\sigma_2)\too\Gamma(\sigma_1)\cross_{\Gamma(\sigma)}\Gamma(\sigma_2)\end{eqnarray} is a bijection.  

Let us first calculate the set $\Gamma(\sigma_1\amalg_\sigma\sigma_2)$.  It is the set of all sections $r$ of the map $\pi'$ in the diagram $$\xymatrix@=30pt{(\sigma_1\amalg_\sigma\sigma_2)^\m1(U)\ullimit\ar[r]\ar[d]^{\pi'}&U\ar[d]^\pi\\C_1\amalg_CC_2\ar@{..>}@/^1pc/[u]^r\ar[r]_-{\sigma_1\amalg_\sigma\sigma_2}&\DT.}$$  To give such a section is to give, for each $c_1\in C_1$ an element of $\pi^\m1(\sigma_1(c_1))$, and for each $c_2\in C_2$ an element of $\pi^\m1(\sigma_2(c_2))$, in such a way that for all $c\in C$, the induced elements in $\pi^\m1(\sigma_i(f_i(c)))$ are the same for $i=1,2$.  This is precisely the data needed for a unique element of the set $\Gamma(\sigma_1)\cross_{\Gamma(\sigma)}\Gamma(\sigma_2)$; this proves the claim that the map in (\ref{eqn:limit schema}) is a bijection.

It now follows that the fiber product of Diagram (\ref{eqn:limit}) is the table $$\tau_1\cross_\tau\tau_2\taking K_1\cross_KK_2\too\Gamma(\sigma_1\amalg_\sigma\sigma_2)$$ obtained by taking the fiber product of sources and targets in (\ref{eqn:limit}), and the induced map between them.

\end{proof}

Proposition \ref{finite limits exist} gives the formula for the join of two tables over a third.  As one sees from the construction, the columns of the join are the union of the columns of the given tables, and the key set is the fiber product of the key sets of the given tables.

\begin{lemma}

Let $\sigma\taking C\to\DT$ denote a simple schema.  The category $\Tables_\sigma$ of tables on $\sigma$ is closed under small limits and colimits.

\end{lemma}

\begin{proof}

The category of sets is closed under small limits and colimits.  To take the limit or colimit of a diagram $X\taking I\to\Tables_\sigma$, simply take the limit or colimit (respectively) of the underlying diagram of key sets -- see Definition \ref{rem:forget}.  This set comes with a natural map to $\Gamma(\sigma)$, and one shows easily that it is the limit or colimit (respectively) of $X$.

\end{proof}

\begin{example}\label{ex:initial and final over sigma}

Let $\sigma\taking C\to\DT$ denote a simple schema.  The initial and final objects in $\Tables_\sigma$ are $\emptyset\to\Gamma(\sigma)$ and $\id_{\Gamma(\sigma)}\taking\Gamma(\sigma)\to\Gamma(\sigma)$, respectively.

\end{example}

\begin{construction}

Let $\tau\taking K\to\Gamma(\sigma)$ be a table with simple schema $\sigma\taking C\to\DT$, and let $C'\ss C$ be a subset of its column set.  There is an induced table $$\tau|_{C'}\taking K\to\Gamma(\sigma|_{C'}).$$  In SQL this construction is called the {\em projection} of $\tau$ onto the subset $C'\ss C$ of columns.

\end{construction}

Using the projection query, one can realize a SELECT query as a limit of databases.  

\begin{construction}

Let us construct the SELECT query.  One begins with a table $\tau\taking K\to\Gamma(\sigma)$ with simple schema $\sigma\taking C\to\DT$, from which to select.  Let $f\taking C'\ss C$ be a subset of its columns, and let $\sigma'=\sigma|_{C'}\taking C'\to\DT$ be the restricted simple schema.  One may select from $\tau$ all records whose restriction to $C'$ is a member of some list.  We encode this list as a table $\tau'\taking K'\to\Gamma(\sigma')$ on $\sigma'$.  

In order to select from $\tau$ all records whose restriction to $C'$ is in the table $\tau'$, take the limit of the diagram $$\xymatrix{K\ar[r]^\tau\ar[d]_{f^*\circ\tau}&\Gamma(\sigma)\ar[d]^{f^*}\\\Gamma(\sigma')\ar[r]^\id&\Gamma(\sigma')\\K'\ar[r]_{\tau'}\ar[u]^{\tau'}&\Gamma(\sigma').\ar[u]_{\id}}$$  This limit is the desired SELECT query.  

\end{construction}

\begin{example}\label{ex:select for tables}

Let $\tau\taking K\to\Gamma(\sigma)$ be the table from Example \ref{ex:table}.  To select all instances for which the first name is Barack, let $C'=\{\tn{`First Name'}\}$.  Let $\tau'$ denote the one-row table $$\begin{tabular}{|l||l|}\hline K'&`First Name'\\\hline\hline k' & Barack\\\hline\end{tabular}$$  Both $\tau$ and $\tau'$ have a canonical map to the terminal table on $C'$, the table with one column (`First Name') and with a row for each element of $\Strings$.  Of course, this terminal table is too big to write down, but we do not need it.  The fiber product is easily computed to be the table $$\begin{tabular}{|l||l|l|l|}\hline K&`First Name'&`Last Name'&`BYear'\\\hline\hline 1&Barack&Obama&1961\\\hline  `foo'&Barack&Obama&1961\\\hline\end{tabular}$$

\end{example}

We conclude this section by a quick remark on the category-theoretic properties of the relational tables.

\begin{remark}

Relations behave much like ordinary tables.  Limits exist in $\Rel$ and $\Rel_\sigma$.  The functor $\Rel\to\Tables$ preserves limits, and the functor $\Rel_\sigma\to\Tables_\sigma$ preserves limits but {\em does not} preserve colimits.

We take the viewpoint that the ``correct" way to take a colimit of a diagram $X\taking I\to\Rel_\sigma$ is to pass to the diagram $I\to\Tables_\sigma$ and take its colimit instead.  This claim, in particular, says that sometimes UNION ALL is more appropriate than UNION is.  Since UNION ALL is not legal in the strict relational database theory (or it would be the same as UNION), our viewpoint could be seen as controversial to purists of the relational model.

\end{remark}

\section{Schemas and databases}\label{sec:schemas and databases}

A relational database is a set of relations, together with a system of keys and foreign keys which link the relations together.  The definition of relations themselves is, of course, quite mathematically precise.  However, the precise way in which these relations are allowed to be linked together is rarely written down as a mathematical structure in its own right, either in research papers or textbooks (we could not find it in \cite{Dat} or \cite{EN}, for example).  For example, ER diagrams are exemplified or even defined, but not as a mathematical object (like relations are).  There are exceptions, such as \cite[2.1]{RW}, but as far as we know, these definitions are not actually the ones used, either by practitioners or by theorists.  

In this section we will define simplicial databases in a rigorous way (see Definition \ref{def:database objects}).  Although examples will be plentiful, they will never stand in for precise definitions.  We will also define morphisms of databases, thus making explicit the idea of ``data-preserving maps."  Providing a precise definition of the category of databases may be useful to database theorists, as well as to people interested in studying mathematical informatics.

\subsection{Schemas}\label{subsec:schemas}

Roughly, a simplicial set is a picture that can be drawn with vertices, edges, solid triangles, solid tetrahedra, and solid ``higher-dimensional tetrahedra."  For any integer $n\geq 0$, an $n$-dimensional solid tetrahedron, or {\em $n$-simplex}, is the ``diagonal triangle" shape in $\RR^{n+1}$ given by the algebraic equation $x_1+x_2+\cdots+x_{n+1}=1$ and the inequalities $x_i\geq 0$ for $1\leq i\leq n+1$.  To draw with these shapes is to connect various tetrahedra together along their faces (or subfaces).  For example, one could connect four triangles together along various faces to obtain an empty tetrahedron, the boundary of the 3-simplex.

Simplicial sets are a fundamental tool in algebraic topology, and are important in many other fields within mathematics, such as combinatorial commutative algebra.   See \cite{Fri} or \cite{GJ} for details.

A database is a system of tables which are connected together via foreign keys.  This information is part of the schema for the database.  In our formulation, we keep track of this information using (something akin to) simplicial sets as our schema.  Tables are connected together when the corresponding simplices are connected.  

We use a slight variant of simplicial sets, which we will define in Definition \ref{def:schema}.  Namely, since columns can only take entries in a given data type, we must keep track of this information.  For this reason, the simplicial sets we use as schema have labeled vertices, where each label is an element of $\DT$.  We do not define schemas exactly this way, however, because a more generalizable way to phrase it may be useful for future generalizations.

\begin{remark}\label{rem:Grandis}

As mentioned in Remark \ref{rem:order}, some prefer the columns of each table in a database to be unordered, whereas we have chosen to consider them as an ordered set.  Simply using symmetric simplicial sets, a variant of simplicial sets in which vertices are unordered, will solve any such issue.  See \cite{Gra} for details on symmetric simplicial sets.

\end{remark}

\begin{definition}\label{def:schema}

Let $\bD$ denote the category of finite ordered sets, let $\pi\taking U\to\DT$ be a type specification, and let $$\mcS\iso(\bD\down\DT)$$ denote the category of simple schema on $\pi$ (see Definition \ref{def:simple schema} and Remark \ref{rem:bD down DT}).  We define {\em the category of schema on $\pi$}, denoted $\Sch^\pi$ to be the category whose objects are functors $X\taking\mcS\op\to\Sets$ and whose morphisms are natural transformations of functors.

Let $X\in\Sch^\pi$ denote a schema.  Given a simple schema $\sigma\taking C\to\DT$, the {\em $\sigma$-simplices} of $X$ are the elements of the set $X(\sigma)$, and we write $X_\sigma$ to denote $X(\sigma)$.

\end{definition}

\begin{remark}\label{rem:sch and fin}

Given a category $\mcC$, the category whose objects are functors $\mcC\op\to\Sets$ and whose morphisms are natural transformations of functors is called {\em the category of presheaves on $\mcC$} and denoted $\Pre(\mcC)$.  It is a common mathematical construction which ``formally adds all colimits to $\mcC$."  That is, $\Pre(\mcC)$ is closed under taking colimits, and for any functor $\mcC\to\mcD$ to a category $\mcD$ which is closed under taking colimits, there is a unique colimit-preserving functor $\Pre(\mcC)\to\mcD$ over $\mcC$.  See, for example, \cite[I.5.4]{MM}.

Thus, we have $\Sch^\pi=\Pre(\mcS^\pi)$.  Since $\mcS^\pi$ signifies the category of ways to set up columns of a tables, $\Pre(\mcS^\pi)$ is the category of ways to glue such things together.

\end{remark}

\begin{remark}\label{rem:labeled}

The category of (augmented) simplicial sets is the category $\Pre(\bD)$.  The only difference between it and $\Pre(\mcS^\pi)\iso\Pre(\bD\down\DT)$ is that each simplex in $\Sch^\pi$ has labeled vertices, whereas simplices in $\Pre(\bD)$ do not.  In the introduction to this section we described simplicial sets in terms of tetrahedra.  After making the necessary modifications, we see that a schema is constructed by gluing together labeled tetrahedra along their faces, where we only allow these tetrahedra to be glued if their labels match.

If $X$ is a schema, we sometimes refer to the simplices of its underlying simplicial set as simplices of $X$.  Thus, the $n$-simplices of $X$ is the union of all $\sigma$-simplices of $X$, where $\sigma\taking C\to\DT$ is a simple schema with cardinality $\card(C)=n+1$.  That is, we write $$X_n=\coprod_{\{\sigma\taking C\to\DT| \card(C)=n+1\}}X_\sigma.$$  There is a classifying map $s\taking X_0=\amalg_{a\in\DT}(X_a)\to\DT$ which sends all of $X_a$ to $a$, for each $a\in\DT$.

\end{remark}

One of the best features of the schema we are presenting here is their geometric nature, as described in the first paragraph of this section.  Unfortunately, Definition \ref{def:schema} does not make the geometry explicit at all.  Hopefully the next few examples will help make it more clear.

\begin{example}\label{ex:simplices}

Let $\sigma\taking C\to\DT$ denote a simple schema.  It naturally defines a schema $X=\Delta^\sigma$ as the functor which sends a simple schema $\sigma'\taking C'\to\DT$ to the set $X_{\sigma'}=\Hom_\mcS(\sigma',\sigma)$.  If $C$ has $n+1$ elements, one visualizes $\Delta^\sigma$ as an $n$-dimensional tetrahedron whose vertices are labeled by elements in the image of $\sigma$.  

This is not just a heuristic: there is a {\em geometric realization} functor $Re:\Sch\to\Top$ which realizes every schema as a topological space in a natural way, and behaves as we have described for simplices $\Delta^\sigma$.

As an example, suppose $C$ has two elements and their images under $\sigma$ are $a,b\in\DT$.  We imagine $\Delta^\sigma$ as a line segment, whose vertices are labeled $a$ and $b$.  If $C'$ has three elements and $\sigma'$ sends two of them to $a$ and one of them to $b$, we imagine $\Delta^{\sigma'}$ as a filled-in triangle, whose vertices are labeled $a,a,$ and $b$.  The figures we have imagined are the images of $\sigma$ and $\sigma'$ under $Re$.

\end{example}

\begin{definition}\label{def:simplex}

Let $\sigma\in\mcS$ denote a simple schema.  The schema $\Delta^\sigma\in\Sch$ defined in Example \ref{ex:simplices} is called {\em the $\sigma$-simplex} and, as a functor $\mcS\op\to\Sets$, is said to be {\em represented by $\sigma$}.

\end{definition}

\begin{example}

We have mentioned that every object in $\Sch^\pi$ can be obtained by gluing together simplices.  This is proven in \cite[2.15.6]{Bor1}.  Let us explain how we would construct the union $X$ of two edges along a common vertex.  Suppose that the common vertex is labeled $b$ and the other vertices are labeled $a$ and $c$.  The schema $X$ is obtained as the colimit of the diagram $$\Delta^{(a,b)}\from\Delta^{(b)}\to\Delta^{(b,c)}$$ taken in $\Sch^\pi$.

We will now write down this schema explicitly as a presheaf on $\mcS^\pi$, i.e. as a functor $X\taking(\bD\down\DT)\op\to\Sets$.  Given $\sigma\taking C\to\DT$, we let $X_\sigma$ be a single element if the image of $\sigma$ is contained in $\{a,b\}$ or contained in $\{b,c\}$.  Otherwise we take $X_\sigma$ to be the empty set.

\end{example}

\begin{example}\label{ex:discrete}

A basic example of a schema is that of a set of labeled vertices with no edges or higher simplices connecting them.  This is obtained as a coproduct of $0$-simplices (see Remark \ref{rem:sch and fin}), and it is called a {\em discrete schema}.  

\end{example}

\subsection{Sheaves on a schema}

\begin{definition}

Let $X\in\Sch^\pi$ denote a schema.  A {\em subschema of $X$} consists of a schema $X'\in\Sch^\pi$ such that for every $\sigma\in\mcS^\pi$ we have $X'_\sigma\ss X_\sigma$.  The subschemas of $X$ form a category $\Sub(X)$, in which there is a morphism $X''\to X'$ in $\Sub(X)$ if and only if $X''$ is a subschema of $X'$.  

\end{definition}

We will soon be discussing colimits in the category $\Sub(X)$.  One should note that $\Sub(X)$ is particularly nice, in that the colimit of any diagram $D\taking I\to\Sub(X)$ is the smallest subschema $X'\ss X$ which contains $D(i)$ for all $i\in I$.  In the language of lattices or locales, one writes $\colim(D)=\bigvee_{i\in I}D(i)$.  See \cite[1.3]{Bor3}.

\begin{definition}

We define a {\em sheaf on $X$} to be a functor $\mcK\taking\Sub(X)\op\to\Sets$ such that, for every diagram $D\taking I\to\Sub(X)$, the natural map $$\mcK(\colim(D))\too\lim(\mcK(D))$$ is an isomorphism.  That is, $\mcK$ must send colimits of subschema to corresponding limits of sets.

A {\em morphism of sheaves on $X$} is a natural transformation of functors $\Sub(X)\op\to\Sets$.  Let $\Shv(X)$ denote the category of sheaves on $X$.

\end{definition}

\begin{remark}\label{rem:sheaves}

Category theory experts will recognize $\Shv(X)$ as the category of sheaves on a certain Grothendieck site (the locale of subobjects of $X$).  It is well known that $\Shv(X)$ is therefore closed under small limits and colimits.  Moreover, there is an adjunction $$\Adjoint{Sh}{\Pre(X)}{\Shv(X)}{}$$ for which the right adjoint is the forgetful functor and the left adjoint is called {\em sheafification}.  Roughly, one sheafifies a presheaf on a schema by replacing its value on each union of simplices by the fiber product of its values on those simplices.  See \cite{MM} for details.

\end{remark}

\begin{example}\label{ex:sheaves}

For any schema $X$, there is an object $\emptyset\in\Sub(X)$, which is the colimit of the empty diagram on $\Sub(X)$.  Hence if $\mcK$ is to be a sheaf on $X$, one must have $\mcK(\emptyset)\iso\{*\}$.

If $X$ is a discrete schema (see Example \ref{ex:discrete}), then $X$ is the coproduct its $0$-simplices.  Thus, if $\mcK$ is to be a sheaf on $X$, we must have $$\mcK(X)=\prod_{x\in X_0}\mcK(x).$$

\end{example}

\begin{example}\label{ex:sheaf on edge}

Suppose that $X\in\Sch^\pi$ is the schema $\Delta^{(`\Str\',`\ZZ\')}$, which looks like this: $$\xymatrix@1{~^{`\Str\'}\!\bullet\hspace{-.1cm}\ar@{-}@<-.2ex>[r]&\hspace{-.2cm}\bullet^{`\ZZ\'}.}$$  The category $\Sub(X)$ is a partially ordered set with five objects: $\emptyset$, $\bullet^{`\Str\'}$,$\bullet^{`\ZZ\'}$, $(\bullet^{`\Str\'},\bullet^{`\ZZ\'})$, and $X$ itself; $\Sub(X)$ has inclusions as morphisms.  

A sheaf $\mcK\in\Shv(X)$ assigns a set to each of these five objects, and functions to each inclusion.  However, by Example \ref{ex:sheaves}, it must assign to $\emptyset$ the terminal set, $\mcK(\emptyset)=\singleton$, and it must assign to $(\bullet^{`\Str\'},\bullet^{`\ZZ\'})$ the product $\mcK(\bullet^{`\Str\'})\cross\mcK(\bullet^{`\ZZ\'})$.  Thus, to specify a sheaf, we need only specify two values, and one morphism, namely $\mcK(X)\to\mcK(\bullet^{`\Str\'})\cross\mcK(\bullet^{`\ZZ\'})$.  

For example we may choose on objects the assignments $\mcK(X)=\{4,cc,10\}$, $\mcK(\bullet^{`\Str\'})=\{1,2\}$, and $\mcK(\bullet^{`\ZZ\'})=\{x,y,z\}$; this implies $\mcK((\bullet^{`\Str\'},\bullet^{`\ZZ\'}))$ is isomorphic to $\{1x,1y,1z,2x,2y,2z\}$.  Any function from $\{4,cc,10\}$ to this six element set, say $4\mapsto 1x, cc\mapsto 2z, 10\mapsto 2z$, defines the restriction maps in our sheaf $\mcK$.  These restriction maps can be thought of as ``foreign keys."

\comment{A diagram depicting $\Sub(X)$.

$$\xymatrix@=10pt{&~^{`\Str'}\!\bullet\hspace{-.1cm}\ar@{-}@<-.2ex>[rr]&\ar[dd]&\hspace{-.2cm}\bullet^{`\ZZ'}\\\\&&\{\bullet^{`\Str'},\bullet^{`\ZZ'}\}\ar[dll]\ar[drr]\\\bullet^{`\Str'}\ar[drr]&&&&\bullet^{`\ZZ}\ar[dll]\\&&\emptyset}\mapsto$$

}

\end{example}

\begin{definition}\label{def:non-deg}

Given a schema $X\in\Sch^\pi$, we have been working with the category $\Sub(X)$ of subschemas of $X$.  There is a related category, called {\em the category of nonempty non-degenerate simple schemas over $X$} and denoted $\ND(X)$, whose objects are monomorphisms $\Delta^\sigma\inj X$ in $\Sch^\pi$, where $\sigma\taking C\to\DT$ is a schema with $C\neq\emptyset$ (see Example \ref{ex:simplices}), and whose morphisms are commutative triangles.  

Every simplex in a schema has a unique underlying non-degenerate simplex (of which it is the degeneracy), so one can define a functor $\ND\taking\Sch^\pi\to\Cat$.

\end{definition}

Since every injection $\Delta^\sigma\inj X$ is in particular a subschema, there is an obvious functor $$\ND(X)\to\Sub(X).$$  This induces an adjunction $\Adjoint{}{\Pre(\ND(X))}{\Pre(\Sub(X)).}{}$  No nontrivial unions exist in $\ND(X)$, so this adjunction becomes $$\Adjoint{L}{\Pre(\ND(X))}{\Shv(\Sub(X)),}{R}$$ where $\Pre(\ND(X))$ is the category of presheaves $\ND(X)\op\to\Sets$.  See \cite[C.1.4.3]{Joh} for more details on this type of construction.

\begin{proposition}\label{prop:non-deg}

Let $X\in\Sch^\pi$ be a schema, and let $\ND(X)$ denote the category of non-degenerate nonempty simple schema over $X$.  The adjunction $$\Adjoint{L}{\Pre(\ND(X))}{\Shv(\Sub(X)),}{R}$$ is an equivalence of categories.

\end{proposition}

\begin{proof}

It is an easy exercise to show that the composition $L\circ R$ is equal to the identity on $\Pre(\ND(X))$ and that $K\circ L$ is canonically isomorphic to the identity on $\Shv(\Sub(X))$.  

\end{proof}

Proposition \ref{prop:non-deg} says that one does not have to worry about sheaves: the category $\Shv(X)$ is equivalent to a category of functors (without ``sheaf" requirements).

\begin{lemma}\label{lemma:f^*}

Let $\pi\taking U\to\DT$ denote a type specification and let $f\taking X\to Y$ denote a morphism of schema on $\pi$.  There is an adjunction $$\Adjoint{f^*}{\Shv(\Sub(Y))}{\Shv(\Sub(X))}{f_*}$$ defined as follows for sheaves $\mcK_X\in\Shv(\Sub(X))$ and $\mcK_Y\in\Shv(\Sub(Y))$.  For any $U\in\Sub(X)$ we take $$f^*\mcK_Y(U)\setto \mcK_Y(f(U)),$$ where $f(U)\in\Sub(Y)$ is the image of $U$ in $Y$.  For any $V\in\Sub(Y)$ we take $$f_*\mcK_X(V)\setto \mcK_X(f^\m1(V)),$$ where $f^\m1(V)$ is the preimage of $V$ in $X$.

\end{lemma}

\begin{proof}

Colimits of presheaves are computed objectwise, and it follows from Proposition \ref{prop:non-deg} that the functor $f^*$, defined above, preserves colimits.  Hence, it suffices to show that for any representable sheaf $rY'=\Hom_{\Sub(Y)}(-,Y')\in\Shv(\Sub(Y))$ and sheaf $T\in\Shv(\Sub(X))$, one has an isomorphism $$\Hom(f^*(rY'),T)\iso^?\Hom(rY',f_*T).$$  To begin, note that for any $U\in\Sub(X)$ one has a chain of natural isomorphisms \begin{align*}f^*(rY')(U)\setto (rY')(f(U))&\iso\Hom_{\Sub(Y)}(f(U),Y')\\&\iso\Hom_{\Sub(X)}(U,f^\m1(Y'))\iso r(f^\m1(Y'))(U).\end{align*}  That is, $f^*(rY')\iso r(f^\m1(Y')).$  By another chain of natural isomorphisms, we have \begin{align*}\Hom(f^*(rY'),T)&\iso\Hom(r(f^\m1(Y')),T)\\&\iso T(f^\m1(Y'))\\&=:f_*T(Y')=\Hom(rY',f_*T).\end{align*}  This proves the lemma.

\end{proof}

\subsection{Simplicial databases}

We think of a schema as a way of organizing the data in a database.  Before we say what a database is, let us give one more example of a schema.  In some sense it will be the fundamental example of a schema; however, it should not really be thought of as a way to organize the data, but as the meaning of the data itself.

\begin{example}\label{ex:universal record}

Let $\pi\taking U\to\DT$ denote a type specification, and let $\mcS=\mcS^\pi$ denote the category of simple schema on $\pi$.  Let $\Gamma^\pi\taking\mcS\op\to\Sets$ denote the functor which assigns to a schema $\sigma\taking C\to\DT$ the set $\Gamma^\pi(\sigma)$ of records on $\sigma$ (see Definition \ref{def:records}).

By Lemma \ref{induced morphisms}, a map $\sigma\to\sigma'$ induces a function $\Gamma^\pi(\sigma')\to\Gamma^\pi(\sigma)$, so $\Gamma^\pi$ is indeed a contravariant functor.  By definition we can consider $\Gamma^\pi$ as a schema on $\pi$ and write $\Gamma^\pi\in\Sch^\pi$.

We call $\Gamma^\pi$ {\em the universal record on $\pi$,} for reasons which will be clear soon.  If the type specification $\pi\taking U\to\DT$ is obvious from context, we may denote $\Gamma^\pi$ simply by $\Gamma$.  

\end{example}

\begin{definition}\label{def:universal sheaf}

Let $\pi\taking U\to\DT$ denote a type specification, let $\Gamma^\pi$ denote the universal record on $\pi$, and let $X\in\Sch^\pi$ denote a schema on $\pi$.  The {\em universal sheaf on $X$ of type $\pi$} is the sheaf $\mcU^\pi$ whose value on a subschema $X'\ss X$ is the set $$\mcU^\pi(X')=\Hom_{\Sch^\pi}(X',\Gamma^\pi).$$  Each element of $\mcU^\pi(X')$ is called a {\em record on $X'$ of type $\pi$}.  If $\pi$ is clear from context, we may write $\mcU$ to denote $\mcU^\pi$.

Now let $X,Y\in\Sch^\pi$ be schema and let $\mcU_X$ and $\mcU_Y$ denote the universal sheaf of type $\pi$ on $X$ and $Y$, respectively.  A map of schema $f\taking Y \to X$ induces a morphism $\mcU_f\taking f^*\mcU_X\to\mcU_Y$ as follows.  Let $Y'\ss Y$ denote an object in $\Sub(Y)$; then composing with $f$ induces a natural map $$f^*\mcU_X(Y')=\Hom_{\Sch^\pi}(f(Y'),\Gamma^\pi)\too\Hom_{\Sch^\pi}(Y',\Gamma^\pi)=\mcU_Y(Y'),$$ which we denote $\mcU_f$; it is similarly defined on morphisms.

\end{definition}

\begin{definition}\label{def:database objects}

Let $\pi\taking U\to\DT$ denote a type specification.  A {\em simplicial database (\tn{or simply} database) of type $\pi$} is a triple $(X,\mcK,\tau)$ where $X\in\Sch^\pi$ is a schema of type $\pi$, $\mcK\in\Shv(X)$ is a sheaf of sets on $\Sub(X)$, and $\tau\taking\mcK\to\mcU_X$ is a morphism of sheaves on $X$ (see Definition \ref{def:universal sheaf}).  We refer to $X$ as {\em the schema}, $\mcK$ as {\em the sheaf of keys}, and $\tau$ as {\em the data} of the database $(X,\mcK,\tau)$.

\end{definition}

\begin{remark}\label{rem:internal keys}

Given a set of ways to measure objects, it often happens that we have several objects with the same measurements.  For example, we may have three green apples, or two 1999 Toyota Corollas.  In relational databases, if two objects have the same attributes, then they are taken to be the same instance.  To keep them distinct, one introduces a unique identifier, an artificial key, which becomes part of the data.  This causes problems with database integration, because the arbitrarily-chosen artificial keys in one database will generally not match with those in another.

In our definition, the keys for the data are kept separate, as the sheaf of sets $\mcK$.  Different names for the keys in no way affect the data itself and therefore do not interfere with database integration.  We say more about this in Section \ref{subsec:duplication}.

\end{remark}

\begin{example}\label{ex:database}

In Example \ref{ex:sheaf on edge}, we wrote down a sheaf $\mcK\in\Shv(X)$ on the schema $$X=\xymatrix@1{~^{`\Str\'}\!\bullet\hspace{-.1cm}\ar@{-}@<-.2ex>[r]&\hspace{-.2cm}\bullet^{`\ZZ\'},}$$ and we will continue to use it in this example.  To specify a database on $X$ of type $\pi$, we must give a morphism $\tau\taking\mcK\to\mcU^\pi$ of sheaves on $X$.

We defined the universal sheaf $\mcU_X$ of type $\pi$ on $X$ in Definition \ref{def:universal sheaf}.  We have \begin{align*}\mcU_X(X)=\mcU_X((\bullet^{`\Str\'},\bullet^{`\ZZ\'}))&=\Str\cross\ZZ\\\mcU_X(\bullet^{`\Str\'})&=\Str\\\mcU_X(\bullet^{`\ZZ\'})&=\ZZ\\\mcU_X(\emptyset)&=\singleton.\end{align*} 

To define a map $\tau\taking\mcK\to\mcU_X$, we must give maps $$\begin{array}{lll}\tau(\bullet^{`\Str\'})\taking\mcK(\bullet^{`\Str\'})\to\mcU_X(\bullet^{`\Str\'}), &&\tau(\bullet^{`\ZZ\'})\taking\mcK(\bullet^{`\ZZ\'})\to\mcU_X(\bullet^{`\ZZ\'})\end{array}$$ and $$\tau(X)\taking\mcK(X)\to\mcU_X(X)$$ that compose correctly with the restriction maps.  We arbitrarily assign $$\begin{array}{lllllll}\tau(1)&=&\tn{Barack}&&\tau(x)&=&1961\\\tau(2)&=&\tn{Michelle}&&\tau(y)&=&1946\\&&&&\tau(z)&=&1964.\end{array}$$  Now $\mcK(X)=\{4,cc,10\}$, and the restriction map sends $4\mapsto 1x$, $cc\mapsto 2z$, and $10\mapsto 2z$.  This forces $\tau(4)=\tn{(Barack; 1961)}$ and $\tau(cc)=\tau(10)=\tn{(Michelle; 1964)}$.  The other values and restriction maps for $\mcK$ are now also forced.

\end{example}

\begin{example}

In Example \ref{ex:database}, we followed the definitions very closely, perhaps to the detriment of the big ideas.  In this example, we write down how the sheaf ``looks" as a collection of tables.  

Let us first change the schema $X$ very slightly, by instead using the schema $\sigma\taking\{\tn{First, BYear}\}\to\DT$, where $\sigma(\tn{First})=\tn{`Str'}$ and $\sigma(\tn{BYear})=`\ZZ$', and now taking $X=\Delta^\sigma$.  The only difference is that we have labeled our columns by more specific attribute names.  We write $\tau(X)\taking\mcK(X)\to\mcU_X(X)$ as the table $$\tau(X)=\begin{tabular}{|l||l|l|}\hline $\mcK(X)$&First&BYear\\\hline\hline4&Barack&1961\\\hline cc&Michelle&1964\\\hline 10&Michelle&1964\\\hline\end{tabular}$$  We write $\tau(\bullet^\tn{First})$ and $\tau(\bullet^\tn{BYear})$ as the tables $$\tau(\bullet^\tn{First})=\begin{tabular}{|l||l|}\hline$\mcK(\bullet^\tn{First})$&First\\\hline\hline 1&Barack\\\hline2&Michelle\\\hline\end{tabular}\hspace{.5in} \tau(\bullet^\tn{BYear})=\begin{tabular}{|l||l|}\hline$\mcK(\bullet^\tn{BYear})$&BYear\\\hline\hline x&1961\\\hline y&1946\\\hline z&1964\\\hline\end{tabular}$$

We can consider the restriction maps $\mcK(X)\to\mcK(\bullet^\tn{First})$ and $\mcK(X)\to\mcK(\bullet^\tn{BYear})$ as foreign keys attached to the $\tau(X)$ table.  The way things are set up, this foreign key information is kept in the restriction maps of the sheaf $\mcK$.  See Example \ref{ex:sheaf on edge}.

\end{example}

\begin{definition}\label{def:database morphisms}

Let $\pi\taking U\to\DT$ denote a type specification, let $\mcX=(X,\mcK_X,\tau_X)$ and $\mcY=(Y,\mcK_Y,\tau_Y)$ denote databases of type $\pi$, and let $\mcU_X$ and $\mcU_Y$ denote the universal sheaf on $X$ and $Y$ (see Definition \ref{def:universal sheaf}).  A {\em morphism of databases}, denoted $$(f,f^\sharp)\taking\mcX\to\mcY,$$ consists of a map $f\taking Y\to X$ of schema (see Definition \ref{def:schema}) and a morphism of sheaves $f^\sharp\taking f^*\mcK_X\to\mcK_Y$ on $Y$ such that the diagram of sheaves \begin{eqnarray}\label{dia:integrity}\xymatrix{f^*\mcK_X\ar[r]^{f^*(\tau_X)}\ar[d]_{f^\sharp}&f^*\mcU_X\ar[d]^{\mcU_f}\\\mcK_Y\ar[r]_{\tau_Y}&\mcU_Y}\end{eqnarray} commutes.

The {\em category of simplicial databases on $\pi$}, whose objects are simplicial databases as defined in Definition \ref{def:database objects} and whose morphisms have just been defined, is denoted $\Data^\pi$, or simply $\Data$ if $\pi$ is understood.  Fixing a schema $X$, the {\em category of databases on $X$}, denoted $\Data_X$, is the category whose objects are databases with schema $X$ and whose morphisms are identity on $X$.  

\end{definition}

\begin{remark}\label{rem:data integrity}

A database is roughly a bunch of tables glued together by foreign key mappings.  A morphism of databases is a way to coherently assign to each table in one database, a table in another database, and a morphism between the two tables.  Recall that a morphism of tables is a ``data-preserving map" (see Definition \ref{def:morphism of tables}, Example \ref{ex:morphism of tables}, and Remark \ref{rem:table integrity}).  Thus, a morphism of databases should be thought of as a coherent system of data-preserving maps.

We might make the following definition.  A {\em morphism without integrity} is a pair $(f,f^\sharp)\taking\mcX\to\mcY$ as above, but {\em without} the requirement that diagram (\ref{dia:integrity}) commute.

\end{remark}

\begin{remark}\label{rem:data_Y}

Let $Y$ be a schema and let $\mcU_Y$ denote the universal database on $Y$.  One can identify $\Data_Y$ with the category $\Shv(Y)_{/\mcU_Y}$ of sheaves over $\mcU_Y$.  Explicitly, this is the category whose objects are arrows $\mcK\to\mcU_Y$ and whose morphisms are commutative triangles.

\end{remark}

\subsection{Relational simplicial databases}\label{Subsec:relational}

In this subsection, we present a category of relational databases as a full subcategory of the category $\Data$ of simplicial databases.  We also give an adjunction which allows one to convert a database in our sense to a relational database in a functorial way.

\begin{definition}

Let $\pi$ denote a type specification.  A simplicial database $\mcX=(X,\mcK,\tau)$ on $\pi$ is called {\em relational} if $\tau\taking\mcK\to\mcU_X$ is a monomorphism of sheaves.  The {\em category of relational simplicial databases}, denoted $\mcRel^\pi$ is the full subcategory of $\Data^\pi$ spanned by the relational simplicial databases.

\end{definition}

Note the precise similarity of this definition with Definition \ref{def:relational table}: the schema $X$ is a gluing together of simple schema $\sigma$, the sheaf $\mcU_X$ evaluated on a simplex $\Delta^\sigma\ss X$ is $\Gamma(\sigma)$, and a monomorphism of sheaves is a morphism which restricts to an injective function on each simplex.

Every function $f\taking A\to B$ between sets has an image $\im(f)\ss B$ and an injection $f^m\taking\im(f)\to B$; similarly, given a schema $X$, every morphism $f\taking \mcA\to\mcB$ of sheaves of sets on $X$ has an image sheaf denoted $\im(f)\ss\mcB$ and a monomorphism of sheaves $f^m\taking\im(f)\to\mcB$.  If $\mcX=(X,\mcK,\tau)$ is a database, we can take the image sheaf $\im(\tau)$ of $\tau\taking\mcK\to\mcU_X$, and the database $(X,\im(\tau),\tau^m)$ will be a relational simplicial database.

\begin{lemma}\label{lemma:adj rel data}

Let $\pi$ denote a type specification.  There is an adjunction $$\Adjoint{}{\Data^\pi}{\mcRel^\pi}{}$$ in which the left adjoint is given by $(X,\mcK,\tau)\mapsto(X,\im(\tau),\tau^m)$ and the right adjoint is the forgetful functor which realizes a relational simplicial database as a simplicial database.

\end{lemma}

\begin{proof}

This is a simple exercise that reduces to the fact that the image functor, which sends the category of sets and functions to the category of sets and injections, is a left adjoint to the forgetful functor.

\end{proof}

Since the forgetful functor $\mcRel^\pi\to\Data^\pi$ is fully faithful, the counit of the adjunction in Lemma \ref{lemma:adj rel data} is the identity functor on $\mcRel^\pi$.  Another way to say this is that one does not lose information when considering a relational database as a simplicial database, but one often does lose information when converting a simplicial database to a relational database.  Strictly ``more information" can be contained in a simplicial database than in a relational database.

\subsection{Tables vs. simplicial databases}

In this last subsection we present the functor $F\taking\Tables\to\Data$ which realizes a table as a simplicial database.  We will also present the ``global table" construction, which roughly takes a database and joins everything together to make one big (unnormalized!) table.

\begin{construction}\label{con:table as database}

Let $\pi\taking U\to\DT$ denote a type specification and $(K,C,\sigma,\tau)$ a table on $\pi$ (see Definition \ref{def:tables}).  Let $X=\Delta^\sigma\in\Sch^\pi$ be the associated schema, let $\mcU_X$ denote the universal database on $X$, and let $\mcK_X$ denote the constant sheaf on $\Sub(X)$ which takes each subschema to the set $K$.  Define $\tau_X\taking\mcK_X\to\mcU_X$ in the unique way such that $\tau_X(X)\taking\mcK_X(X)\to\mcU_X(X)$ is the function $\tau\taking K\to\Gamma(\sigma)$.  We are ready to assign $$F((K,C,\sigma,\tau))\setto (X,\mcK_X,\tau_X).$$

Given a map of tables $\varphi\taking(K_1,C_1,\sigma_1,\tau_1)\to(K_2,C_2,\sigma_2,\tau_2)$, we will now show that there is a canonical map of simplicial databases $(X_1,\mcK_1,\tau_1)\to(X_2,\mcK_2,\tau_2)$.  Recall from Definition \ref{def:morphism of tables} that $\varphi=(g,f)$ where $g\taking K_1\to K_2$ is a function and $f\taking\sigma_2\to\sigma_1$ is a morphism of simple schema such that Diagram (\ref{dia:morphism of tables}), rewritten for the readers convenience here: $$\xymatrix{K_1\ar[r]^{\tau_1}\ar[d]_g&\Gamma(\sigma_1)\ar[d]^{f^*}\\K_2\ar[r]_{\tau_2}&\Gamma(\sigma_2),}$$ commutes.

The morphism $f\taking\sigma_2\to\sigma_1$ of simple schema induces a morphism $\Delta^{\sigma_2}\to\Delta^{\sigma_1}$ of schema, i.e. a map $f\taking X_2\to X_1$.  The sheaf $f^*\mcK_1$ on $X_2$ is the constant sheaf with value $K_1$, so $g$ gives a map $f^\sharp\taking f^*\mcK_1\to\mcK_2$.  We will skip some details, but one can easily show that the commutativity of the Diagram (\ref{dia:integrity}) is equivalent to the commutativity of Diagram (\ref{dia:morphism of tables}), completing the construction.

\end{construction}

We can also extract a single table from a simplicial database, by looking at its global sections.  This requires a functor called $f_+$ defined in Section \ref{subsec:changing the schema}.  We include the construction here, rather than later, in order to keep like topics together, and conclude nicely with Remark \ref{rem:adj tables data}.

\begin{construction}\label{con:database as table}

Let $\mcX=(X,\mcK,\tau)$ denote a simplicial database.  Recall from Remark \ref{rem:labeled} that there is an induced classification map $s\taking X_0\to\DT$.  Assuming that $X$ has finitely many vertices, we can construct a table whose simple schema is $s$.  

To do so, we need only note that there is a unique map of schema $f\taking X\to\Delta^s$.  Indeed, given any simplex in $X$, its set of vertices classifies a unique simplex in $\Delta^s$; this defines $f$.  If we write $K=\mcK(X)=f_+\mcK(\Delta^s)$ and $t=f_+\tau_X(\Delta^s)\taking K\to\Gamma(s)$, then we are ready to construct the table $$(K,X_0,s,t)\in\Tables.$$  Its columns are given by the vertices $X_0$ of $X$; its rows are difficult to describe in general, but in specific cases are quite sensible.

\end{construction}

\begin{remark}\label{rem:adj tables data}

It is not hard to show that the two above constructions establish an adjunction $$\Adjoint{}{\Tables}{\Data}{}$$  Given a database $\mcX$, the table obtained by the right adjoint will be called the {\em global table on $\mcX$.}

\end{remark}

\section{Constructions and formal properties of Simplicial Databases}\label{sec:constructions for databases}

The point of the formalism in Section \ref{sec:schemas and databases} is to find a language in which to describe databases such that the typical operations performed when working with databases are sensible in the language.  In other words, queries of databases should make sense as categorical constructions, as they did in Section \ref{sec:constructions for tables} for tables.  

\subsection{Changing the schema}\label{subsec:changing the schema}

Let us begin with some ways that one can import data from one schema into another.  In Lemma \ref{lemma:f^*} we discussed the adjunction \begin{eqnarray}\label{dia:image sheaves}\Adjoint{f^*}{\Shv(\Sub(Y))}{\Shv(\Sub(X))}{f_*}\end{eqnarray} induced by a map of schema $f\taking Y\to X$.  Given a database $\mcX=(X,\mcK_X,\tau_X)$ on $X$ there is an induced database $(Y,f^*\mcK_X,\mcU_f\circ (f^*\tau_X))$, denoted $f^*\mcX$; see Definition \ref{def:universal sheaf} and refer to the diagram $$\xymatrix{f^*\mcK_X\ar[r]^{f^*\tau_X}\ar[dr]&f^*\mcU_X\ar[d]^{\mcU_f}\\&\mcU_Y.}$$  

A slightly more complicated construction creates a database on $X$ from a database $\mcY=(Y,\mcK_Y,\tau_Y)$ on $Y$ and a map of schema $f\taking Y\to X$.  By the adjunction (\ref{dia:image sheaves}), we have the diagram \begin{eqnarray}\label{dia:problem}\xymatrix{&\mcU_X\ar[d]\\f_*\mcK_Y\ar[r]_{f_*\tau_Y}&f_*\mcU_Y,}\end{eqnarray} but since there is no canonical map $f_*\mcK_Y\to\mcU_X$, we have not yet constructed a database on $X$.

To do so, let $f_+(\mcK_Y)$ denote the limit of Diagram (\ref{dia:problem}).  This sheaf comes with a canonical map to $\mcU_X$, which we denote $f_+\tau_Y\taking f_+\mcK_Y\to\mcU_X$.  The triple $$(X,f_+\mcK_Y,f_+\tau_Y)$$ is a database on $X$, which we denote $f_+\mcY$.  

\begin{proposition}

Let $\pi$ denote a type specification, and let $f\taking Y\to X$ be a morphism of schema of type $\pi$.  The functors $f^*$ and $f_+$ define an adjunction $$\Adjoint{f^*}{\Data_X}{\Data_Y.}{f_+}$$

\end{proposition}

\begin{proof}

Let $\mcX=(X,\mcK_X,\tau_X)$ and $\mcY=(Y,\mcK_Y,\tau_Y)$ be databases.  Giving a morphism $f^*\mcX\to\mcY$ of databases over $Y$ amounts to a giving a map $\alpha$ of sheaves making the diagram $$\xymatrix{f^*\mcK_X\ar[r]^{f^*\tau_X}\ar@{-->}[d]_\alpha&f^*\mcU_X\ar[d]^{\mcU_f}\\\mcK_Y\ar[r]_{\tau_Y}&\mcU_Y}$$ commute.  By the adjunction (\ref{dia:image sheaves}) this diagram is equivalent to the diagram $$\xymatrix{\mcK_X\ar[r]^{\tau_X}\ar@{-->}[d]_\alpha&\mcU_X\ar[d]^{\mcU_f}\\f_*\mcK_Y\ar[r]_{f_*\tau_Y}&f_*\mcU_Y,}$$ by Lemma \ref{lemma:f^*}.  Supplying a morphism $\alpha$ making this diagram commute is equivalent to supplying a morphism $\mcK_X\to f_+\mcK_Y$ over $\mcU_X$, because $f_+\mcK_Y$ is the limit of Diagram \ref{dia:problem}.  The proof now follows from Remark \ref{rem:data_Y}.

\end{proof}

\begin{definition}\label{def:image databases}

Let $\pi$ denote a type specification, and let $f\taking Y\to X$ be a morphism of schema of type $\pi$.  The functor $f^*\taking\Data_X\to\Data_Y$, defined above, is called the {\em pullback functor}, and the functor $f_+\taking\Data_Y\to\Data_X$, defined above, is called the {\em push-forward functor}.  

Given a sheaf of sets $\mcK_X$ on $X$, we also refer to $f^*\mcK_X\in\Shv(Y)$ as {\em the pullback of $\mcK_X$}, and given a sheaf of sets $\mcK_Y$ on $Y$, we also refer to $f_+\mcK_Y\in\Shv(X)$ as {\em the push-forward of $\mcK_Y$}.

\end{definition}

\begin{example}

Let $X$ and $Y$ be the schema $$X\setto \xymatrix@1{~^{`\Str\'}\!\bullet\hspace{-.1cm}\ar@{-}@<-.2ex>[r]&\hspace{-.2cm}\bullet^{`\ZZ\'},} \hspace{.5in}Y\setto \xymatrix@1{~^{`\Str\'}\!\bullet\hspace{-.1cm}\ar@{-}@<-.2ex>[r]&\hspace{-.2cm}\bullet^{`\Str\'},}$$ and let $f\taking Y\to X$ be the unique morphism of schema between them.  

By Remark \ref{rem:data_Y}, a database on $X$ is given by a morphism of sheaves $\tau_X\taking\mcK_X\to\mcU_X$, for some sheaf of sets $\mcK_X$.  We roughly think of it as a table of strings and integers, with some values not filled in.  (In fact, $\tau_X$ has more information because, for example, two keys in $\mcK(X)$ might be sent to the same key in $\mcK(\bullet^{`\Str\'})$).

The pullback database $f^*\tau_X\taking f^*\mcK_X\to\mcU_Y$ is degenerate in the sense that every row has the same string repeated in two columns.  In some sense, this is to be expected.

Now suppose that $\tau_Y\taking\mcK_Y\to\mcU_Y$ is a database on $Y$.  We roughly think of it as a table whose rows are pairs of strings.  The push-forward $f_+\tau_Y$ consists of three tables: one has two columns (strings and integers) and the other two just have one column.  The one column table of integers $f_+\tau_Y(\bullet^{`\ZZ\'})$ is empty.  The one column table of strings $f_+\tau_Y(\bullet^{`\Str\'})$ consists of those strings $S$ for which there is a row in $\tau_Y(Y)$ consisting of a repeated string $(S,S)$.  Finally, the two column table $f_+\tau_Y(X)$ consists of an element $(S,n)$ for every row $S$ in the one-column table of strings and every integer $n\in\ZZ$. 

One sees that by this example that if $f\taking Y\to X$ is not surjective, then the pushforward functor $f_+$ results in huge tables.  It is not meant to be implemented as a hash table but as a theoretical construct.

\end{example}

Given a map of schemas $f\taking Y\to X$, there is one more important way to send a database on $X$ to a database on $Y$, but only if $f$ is a monomorphism of schema.  A monomorphism of schema corresponds to the relationship often known as ``is a", in which every object of type $x$ ``is an" object of type $y$.  In this situation, there is a functor which takes as input a database of $y$'s, and produces as output a database of $x$'s with all of the $y$-information filled in, but nothing else.  The functor that accomplishes this task is denoted $f_!\taking\Data_Y\to\Data_X$ and is called ``extension by $\emptyset$," meaning that on every simplex in $X$ that is not in $f(Y)$, the value of the sheaf there is an empty table.  

To define $f_!$ rigorously, we first notice that $f^*\taking\Shv(X)\to\Shv(Y)$ not only has a right adjoint ($f_*$), but a left adjoint as well, which we also denote $f_!\taking\Shv(Y)\to\Shv(X)$.  If $f$ is a monomorphism, then every subschema $Y'\ss Y$ is sent to a subschema $f(Y')\ss X$.  

Let us define $f_!\mcU_Y$ and its canonical map to $\mcU_X$.  Every subschema $X'\ss X$ is either of the form $X'=f(Y')$ or not.  If so, we set $f_!\mcU_Y(X')=\mcU_Y(Y')=\mcU_X(X')$.  If not, we set $f_!\mcU_Y(X')=\emptyset$.  There is a canonical map $a_f\taking f_!\mcU_Y\to\mcU_X$ which is the identity map on $X'=f(Y')$ and which is $\emptyset\to\mcU_X(X')$ when $X'\not\in\im(f)$.

Now that we have a canonical map $a_f\taking f_!\mcU_Y\to\mcU_X$ in the case that $f\taking Y\to X$ is an inclusion, we can define $f_!\taking\Data_Y\to\Data_X$ to be given by $$f_!(Y,\mcK_Y,\tau_Y)\setto (X,f_!\mcK_Y,a_f\circ\tau_Y).$$  The functor $f_!$ is left adjoint to the functor $f^*\taking\Data_X\to\Data_Y$ (but $f_!$ is defined only when $f\taking Y\to X$ is an injection.)

\subsection{Nulls}\label{subsec:nulls}

Nulls do not conform with the mathematical logic that underlies the strict theoretical foundation of relational databases.  They are easy enough to deal with, however, by use of foreign keys.  That is, for each column $c\in C$ of a schema $\sigma\taking C\to\DT$ for which a table may contain a null, one creates a new schema $\sigma'$ on columns $C'=C-\{c\}$.  By an easy use of foreign keys, one considers objects classified by $\sigma$ to be also classified by $\sigma'$.  This is a way to get around the problem of nulls.  Other approaches can be found in \cite{JR}.

The same technique is done (automatically) in simplicial databases.  Over a simplex $\Delta^\sigma$, one puts objects for which the value on each column is known.  If the value on some set of columns is unknown for a certain object, it is represented as a record on the subsimplex for which it is total.  

If one so desired, he or she could implement simplicial databases so that local sections of the database (records over subschema) appeared as global sections of the database (records over the whole schema) by putting the value ``Null" in appropriate places.  From our perspective it is preferable just to leave local data as local data and not try to promote it to global data, at least for theoretical purposes.

\subsection{Duplicate records}\label{subsec:duplication}

SQL allows for a table to have the same record in two different rows.  Therefore, tables are not relations and SQL does not strictly implement relational databases.  One could argue that SQL is ``wrong" in not conforming to the theory (see \cite[p. 14]{Dat}), but perhaps the pure relational theory is overly strict; this is the position we take.  

Simplicial database allow for duplicate entries.  This should not be threatening because internal keys ensure the integrity of the data.  If $\Gamma=A\cross B\cross C$, then relations on this simple schema are subsets $K\ss\Gamma$.  In the theory of simplicial databases, we allow non-injective functions $\tau\taking K\to\Gamma$, called tables.

Philosophically, we see the relational model as ``confusing the object with its attributes."  A schema, or set of attributes, gives a set of ways to measure a collection of objects.  It is entirely possible that two objects in that collection could have the same measurements according to the schema.  In the relational model, these two objects would be {\em identified} in the sense that only one row of the table would be representing both.  From now on, the database and its users will have no choice but to consider these objects to be the same.

The only alternative to this is to introduce arbitrary identifiers.  These artificial keys are not part of the data being measured about the objects.  In our view, it is best to keep these arbitrary identifiers ``internal" to the database management system.  Among several advantages, the most obvious is database integration, in which it is important to know what aspects of the data are ``measured" and invariant, and what aspects are contrived.  We will say more about this in Section \ref{subsubsec:database integration}.

\comment{
The relational model tends to identify objects with a certain list of their attributes.  No two objects are allowed to have the same attributes in the relational model.  One can claim that in reality, no two objects have precisely the same attributes, and this is true.  However a schema does not consider every possible attribute but a small and pre-defined set of attributes.  It is quite possible that two objects will look the same to that schema. 

In the simplicial database model (or any model in which one allows non-injective functions $\tau$ as tables), the object is not confused with its list of attributes.  That is, different objects can ``look the same" to a given schema.  The schema represents attributes with which to distinguish between objects, and the relational model demands perfection on this front: if two objects are different then they must be distinguishable by the chosen schema.  Our simplicial model does not demand perfection in this sense, but loses no mathematical rigor or power; it simply gains flexibility.  As shown in Lemma \ref{lemma:adj rel data}, every relational database can be considered as a simplicial database in a functorial way, but there are strictly more simplicial databases because they allow for duplicate entries.}

\subsection{Limits and colimits of databases}

We will see shortly that limits and colimits taken in the category of simplicial databases have meaning in terms of the general theory of databases, such as joins and unions.  

\begin{theorem}\label{thm:colimits and limits}

Let $\pi\taking U\to\DT$ denote a type specification.  The category $\Data^\pi$ of databases of type $\pi$ is closed under taking small colimits and small limits.

\end{theorem}

\begin{proof}

Let $I$ denote a small category and let $\mcX\taking I\to\Data$ denote an $I$-shaped diagram in $\Data=\Data^\pi$.  There is a functor $\Data\to\Sch\op$ taking a database $(A,\mcK_A,\tau_A)$ to its underlying schema $A$, and composing this functor with $\mcX$ gives a functor which we denote $X\taking I\to\Sch\op$.  For an object $i\in I$, we denote the database $\mcX(i)$ by $\mcX_i$ and write $$\mcX_i=(X_i,\mcK_i,\tau_i).$$

To define the colimit (respectively limit) of the diagram $\mcX$, we must first specify its schema.  Since $\Sch=\Pre(\mcS)$, where $\mcS$ is the category of simple schema (see Definition \ref{def:category of schema}), it is closed under colimits and limits (\cite[p. 22]{MM}); hence so is $\Sch\op$.  Let $C=\colim(X)$ (resp. $L=\lim(X)$) denote the colimit (resp. limit) of the diagram $X\taking I\to\Sch\op$.  Let $\mcU_C$ and $\mcU_L$ denote the universal databases on $C$ and $L$, respectively.

As a colimit in $\Sch\op$, the schema $C$ comes equipped with morphisms in $c_i\taking C\to X_i$ in $\Sch$, for each $i\in I$, making the appropriate diagrams commute.  There is a pullback sheaf $c_i^*\tau\taking c_i^*\mcK_i\to\mcU_C$.  If $f\taking i\to j$ is a morphism in $I$, then the map $X_j\to X_i$ in $\Sch$ induces a morphism $$c_i^*\mcK_i\to c_j^*\mcK_j$$ of pullback sheaves over $\mcU_C$ on $C$.  Let $c^*\taking I\to\Shv(C)_{/\mcU_C}$ denote the $I$-shaped diagram of these pullback sheaves over $\mcU_C$.  Define $\tau_C\taking\mcK_C\to\mcU_C$ to be the colimit of this diagram.  Then the database $$\mcC=(C,\mcK_C,\tau_C)$$ is our candidate for the colimit of the diagram $\mcX$.  It is a matter of tracing through the construction to show that $\mcC$ has the necessary universal property.

Defining the limit of $\mcX$ is similar.  As a limit in $\Sch\op$, the schema $L$ comes equipped with morphisms $\ell_i\taking X_i\to L$ in $\Sch$, for each $i\in I$, making the appropriate diagrams commute.  There is a push-forward sheaf $(\ell_i)_+\mcK_i$ on $L$, which comes equipped with a map $(\ell_i)_+\tau\taking(\ell_i)_+\mcK_i\to\mcU_L$.  If $f\taking i\to j$ is a morphism in $I$, then the map $X_j\to X_i$ in $\Sch$ induces a morphism $$(\ell_i)_+\mcK_i\to(\ell_j)_+\mcK_j$$ of push-forward sheaves over $\mcU_L$ on $L$.  Let $(\ell_+)\taking I\to\Shv(L)_{/\mcU_L}$ denote the $I$-shaped diagram of these push-forward sheaves over $\mcU_L$.  Define $\tau_L\taking\mcK_L\to\mcU_L$ to be the limit of this diagram.  Then the database $$\mcL=(L,\mcK_L,\tau_L)$$ is our candidate for the limit of the diagram $\mcX$.  Again, it is a matter of tracing through the construction to show that $\mcL$ has the necessary universal property.  

This completes the proof.

\end{proof}

\begin{remark}

The final object in the category $\Data^\pi$ of databases on $\pi\taking U\to\DT$ is the empty database (with empty schema and trivial sheaf).  The initial object $(X,\mcK,\tau)$ in $\Data^\pi$ has, as its schema $X$, a single $n$-simplex for every map $\sigma\taking\{0,1,\ldots,n\}\to\DT$; the sheaf is $\mcK=\mcU_X$, and the map $\tau\taking\mcU_X\to\mcU_X$ is the identity.

If one knows the \Cech nerve construction, one can realize the initial object in those terms, by applying the \Cech nerve functor to $\pi\taking U\to\DT$.  See \cite[3.1]{Spi} for details.

\end{remark}

\begin{corollary}\label{cor:colimits and limits}

Let $X\in\Sch$ be a schema and let $\Data_X$ denote the category of databases with schema $X$ and with morphisms which restrict to the identity on $X$.  Colimits and limits exist in $\Data_X$; in particular $\Data_X$ has an initial object and a final object.

\end{corollary}

\begin{proof}

Given a non-empty diagram which restricts to the identity on a certain schema $X$, one sees by the construction of limits and colimits in the proof of Theorem \ref{thm:colimits and limits} that the limit and the colimit of that diagram will also have schema $X$.  

The limit (respectively the colimit) of the empty diagram in $\Data_X$, if it exists, is the final (resp. initial) object in $\Data_X$; we must show it does exist.  One immediately sees that the final object is $(X,\mcU_X,\id_{\mcU_X})$, and the initial object is $(X,\emptyset,\emptyset\to\mcU_X)$, where $\emptyset$ here denotes the sheaf on $X$ whose value is constantly the empty set, and where $\emptyset\to\mcU_X$ is the unique morphism of sheaves.

\end{proof}

\subsection{Projections}\label{subsec:projections}

This query is built into the theory of simplicial databases.  Given a database $(X,\mcK,\tau)$ and a subschema $X'\ss X$, we have the database $(X',\mcK|_{X'}\tau|_{X'})$ given by restricting the sheaf $\mcK$ and the map of sheaves $\tau\taking\mcK\to\mcU$ to the subschema $X'$.  One can view it as a table using Construction \ref{con:database as table}.

\subsection{Unions and insertions}

Given two databases with the same schema, one can apply the UNION query.  To do so, one keeps the same columns but takes the union of the rows.  An insertion is a special kind of union; namely it is a union of two databases on the same schema, where one of the databases consists only of a single row.

We have a few more options in simplicial databases than one does in relational databases; these differences are analogous to the difference between the UNION query and the UNION ALL query in SQL.  That is, since we allow duplicate entries (see Section \ref{subsec:duplication}), the user can decide when an object in one database is the same as an object with the same attributes stored in another database and when it is different.  Let us make all of this precise.

We can represent unions, insertions, and more by taking colimits of various diagrams of databases.  Let $\mcX=(X,\mcK,\tau)$ denote a simplicial database, and let $\mcX'=(X,\mcK',\tau')$ be a database with the same schema, $X$.  Both receive a map from the initial database on $X$, and the coproduct will be $(X,\mcK\amalg\mcK',\tau\amalg\tau')$ as desired.  (See the proof of Theorem \ref{thm:colimits and limits} for details on the colimit construction.)

The above construction gives a UNION ALL query: duplicated tuples will remain distinct.  There are two ways of having that not be the case.  The first is to simply eliminate the duplicates by converting the database to a relational database; see Lemma \ref{lemma:adj rel data}.  However, this may result in information loss if there really were two entities with the same attributes, because these duplicates will be eliminated.

The other way can occur if the user has more information about which instances in the first database correspond to instances in the second database.  This can be accomplished by having a third database $\mcX''=(X,\mcK'',\tau'')$ and maps from it to $\mcX$ and $\mcX'$.  The colimit of this diagram, $(X,\mcK\amalg_{\mcK''}\mcK',\tau\amalg_{\tau''}\tau')$, will be the union of the records in $\mcX$ with those in $\mcX'$, and will identify two records if they agree in $\mcX''$.

As mentioned above, inserting a row is a special case of taking the union of databases.

We can take much more general colimits than those mentioned above, all of which were constant in the schema.  These constructions appear to be new; perhaps they can provide useful ways to analyze and assemble data.

\subsection{Join}\label{subsec:join}

Two databases can be joined together by specifying a common sub-schema of each and ``gluing together" along that sub-schema.  If no common sub-schema is mentioned we take the initial schema, which is empty, and join along that; the result is called the natural join.  The concept of gluing is rigorously formulated as taking limits of certain diagrams in $\Sch\op$; thus the point we are making is that joining databases in the usual sense can be accomplished by taking limits in the category of simplicial databases.  Let us make all of this precise.

Recall from Theorem \ref{thm:colimits and limits} that the limit of the diagram of databases $$(X_1,\mcK_1,\tau_1)\too(X,\mcK,\tau)\fromm(X_2,\mcK_2,\tau_2)$$ has schema $X'=X_1\amalg_XX_2$.  This induces a diagram $$\xymatrix{X\ar[r]\ar[d]&X_1\ar[d]\\X_2\ar[r]&X'\lrlimit}$$ in $\Sch$.  We can thus push-forward $\mcK_1$, $\mcK$, and $\mcK_2$ to $X'$ and get a diagram of push-forward sheaves there (see Definition \ref{def:image databases}), all naturally mapping to $\mcU_{X'}$.  For typographical reasons, we leave out the fact that these are push-forwards and write the diagram $\mcK_1\to\mcK\from\mcK_2$ over $\mcU_{X'}$.  We are ready to write the limit database as $$(X_1\amalg_XX_2,\mcK_1\cross_\mcK\mcK_2,\tau'),$$ where $\tau'\taking\mcK_1\cross_\mcK\mcK_2\to\mcU_{X'}$ is the structure map.

\begin{example}

Suppose we have the two schemas pictured here: $$X_1\setto \xymatrix@1{~^\tn{`First'}\!\bullet\hspace{-.1cm}\ar@{-}@<-.2ex>[r]&\hspace{-.2cm}\bullet^\tn{`Last'}}, \hspace{.5in}X_2\setto \xymatrix@1{~^\tn{`L.Name'}\!\bullet\hspace{-.1cm}\ar@{-}@<-.2ex>[r]&\hspace{-.2cm}\bullet^\tn{`BYear'},}$$ and wish to join them together by equating `Last' with `L.Name' (both of which have the same data type, namely $\Str$).  To do so, we use the schema $X=\bullet^{`\Str'}$, which maps to each of $X_1$ and $X_2$ in an obvious way.  

Now given any databases $\mcX_1=(X_1,\mcK_1,\tau_1)$ and $\mcX_2=(X_2,\mcK_2,\tau_2)$ on $X_1$ and $X_2$, we can join them by taking the limit of the solid arrow diagram $$\xymatrix{\mcX_1\cross_\mcX\mcX_2\ar@{-->}[r]\ar@{-->}[d]&\mcX_2\ar[d]\\\mcX_1\ar[r]&\mcX}$$ where $\mcX=(X,\mcU_X,\id_{\mcU_X})$ is the final database on $X$.  The schema of the resulting database is

\begin{figure}[htbp]
\setlength{\unitlength}{3947sp}
\includegraphics{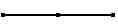}
\put(-2200,50){`First'}
\put(-1550,50){`Last'=`LName'}
\put(-200,50){`BYear'}
\end{figure}

This {\em does not} represent a table with three columns, but two tables, each with two columns, and each projecting to a common 1-column table.  However, its global table does have three columns (see Remark \ref{rem:adj tables data}).  Its records are those triples of the form (First,Last,BYear) for which there is a (First,Last) pair in $\mcX_1$ and a (Last,BYear) pair in $\mcX_2$ with matching values of Last.  This is indeed their join.

\end{example}

\begin{remark}

The ``join" we are working with here could be thought of as a combination of equi-join and outer join.  Because databases are sheaves on a schema, they do not have just one table but a system of tables, and the idea of nulls is built into the theory (see Section \ref{subsec:nulls}).  

More precisely, if $\mcX_1\to\mcX\from\mcX_2$ is a diagram of databases, the limit $\mcX'$ represents the join of $\mcX_1$ and $\mcX_2$ along a shared set of columns (those of $\mcX$).  Its schema is roughly the union of the schemas of $\mcX_1$ and $\mcX_2$.  Its global table will be the equi-join of the global tables for $\mcX_1$ and $\mcX_2$.  

The point of this remark, however, is that the new table $\mcX'$ does not only contain global information, but local information as well.  Much of the data of $\mcX_1$ (respectively $\mcX_2$) is preserved upon passage to $\mcX'$, and that which cannot be extended to global data could still be viewed globally if one uses Null values.  It is in this sense that colimits in $\DB$ are related to outer joins.

\end{remark}

When joining databases together, one first chooses a set $C$ of columns to equate.  When two distinct objects have the same $C$-attributes, then the join is ``lossy" in the sense that there will be false information in the join.  To remedy this, one must be careful to distinguish between objects, even when considered only in terms of $C$.  The following example will hopefully make this more clear.

\begin{example}

Suppose one wants to join the following two tables: $$\xymatrix@=12pt{\hbox{\begin{tabular}{|l||l|l|}\hline$\tau_1$&Title&LastName\\\hline\hline 1&Dr.&Marx\\\hline 2&Mr.&Marx\\\hline\end{tabular}}&&\hbox{\begin{tabular}{|l||l|l|}\hline$\tau_2$&FirstName&LastName\\\hline\hline A&Karl&Marx\\\hline B&Groucho&Marx\\\hline\end{tabular}}}$$ The outcome will be the following table: $$\begin{tabular}{|l|l|l|}\hline Title&FirstName&LastName\\\hline\hline Dr.&Karl&Marx\\\hline Dr.&Groucho&Marx\\\hline Mr.&Karl&Marx\\\hline Mr.&Groucho&Marx\\\hline\end{tabular}$$  This table has four entries, two of which are ``accurate," in that they describe real instances, and two of which are not.  This occurs because the relational database cannot distinguish between the two instances of the last name Marx.  

Achieving a lossless join is easy, when databases are allowed to have duplicate entries with the same attributes.  Consider the table $$\begin{tabular}{|l||l|}\hline$\tau$&LastName\\\hline\hline x&Marx\\\hline y&Marx\\\hline\end{tabular}$$ which accepts maps from both $\tau_1$ and $\tau_2$ by sending both $1$ and $A$ to $x$, and sending both $2$ and $B$ to $y$ (see Definition \ref{def:morphism of tables}).  The limit of this diagram is the table $$\begin{tabular}{|l|l|l|}\hline Title&FirstName&LastName\\\hline\hline Dr.&Karl&Marx\\\hline Mr.&Groucho&Marx\\\hline\end{tabular}$$ as desired.

\end{example}

In the example above, the table $\tau$ has two instances of the same string.  This is not superfluous because there are two people named Marx.  They are differentiated by their internal keys, but not by their attributes.  Keeping distinct objects distinct, even if they have the same attributes is very useful in practice.  It not only allows for lossless joins, but it is well-suited for database integration as well.

\subsection{Select}\label{subsec:select} 

In Example \ref{ex:select for tables}, we selected from a table $\tau$ with columns $C=\{\tn{`First Name', `Last Name', `BYear'}\}$ all instances for which the value of `First Name' was ``Barack."  This was computed as follows.  First, we made a table $\tau'$ whose column set $C'$ consisted of a single element, labeled `First Name', and filled in $\tau'$ with a single entry, `Barack'.  We might call this table the {\em selection table}.  The SELECT operation was performed by taking the fiber product $\tau\to\id_{C'}\from\tau'$, where $\id_{C'}$ denotes the table of all possible values of `First Name'.

Performing SELECT operations in a general simplicial database has the same flavor, in that it is always computed as a certain kind of fiber product.  Denote the database from which we are selecting as $\mcX=(X,\mcK_X,\tau_X)$, let $S\ss X$ denote a subschema and $\mcS=(S,\mcK_S,\tau_S)$ a relational table on $S$, to serve as the selection table.  That is, we will be selecting from $X$ all instances that have the designated $S$-attributes.  Finally, we let $1_\mcS=(S,\mcU_S,\id_{\mcU_S})$ denote the final database on the schema $S$.  The fiber product $\mcX_\mcS$ in the diagram $$\xymatrix{\mcX_\mcS\ar[r]\ar[d]\ullimit&\mcS\ar[d]\\\mcX\ar[r]&1_\mcS}$$ is the desired result.

\subsection{Deletions}\label{subsec:delete}

Deletion can be subtle.  If one deletes entries over a subschema, the action must ``cascade" up the hierarchy, deleting entries in larger schemas when they refer or point to the deleted entries.  To that end, we define the following construction.  

\begin{definition}

Suppose given a schema $X$ and a subsheaf $\mcK_1\ss\mcK$ on $X$.  Let $\ol{\mcK_1}\ss\mcK$ denote the presheaf on $X$ with $$\ol{\mcK_1}(X')\setto \{r\in\mcK(X')|\exists X''\ss X', X''\neq\emptyset, r_{X''}\in\mcK_1(X'')\}$$ for subschema $X'\in\Sub(X)$.  Here $r_{X''}$ denotes the image of $r$ under the restriction map $\mcK(X')\to\mcK(X'')$.  We call $\ol{\mcK_1}$ the {\em closure of $\mcK_1$ in $\mcK$}.

\end{definition}

Suppose now we want to delete all entries of a given type from a database.  More concretely, suppose $\mcX=(X,\mcK_X,\tau_X)$ is a database with schema $X$, that $i\taking S\ss X$ is a subschema, and that $\mcS=(S,\mcK_S,\tau_S)$ is a relational database of objects of this subtype, all of which we would like to delete from $X$.  As explained in Section \ref{subsec:select}, we can select the rows of $\mcX$ of the type specified by $\mcS$ by defining $\mcX_S$ to be the limit as in the diagram $$\xymatrix{\mcX_S\ar[r]\ar[d]\ullimit&\mcS\ar[d]\\\mcX\ar[r]&(S,\mcU_S,\id_{\mcU_S}).}$$  We know that $\mcX_\mcS$ has schema $X=X\amalg_SS$ and we momentarily invent notation and write $\mcX_\mcS=(X,\mcK_{\mcS\ss\mcX},\tau_{\mcS\ss\mcX})$.  

The map $\mcX_\mcS\to\mcX$ defines an inclusion of sheaves $\mcK_{\mcS\ss\mcX}\ss\mcK_X$ on $X$, and we take its closure $\ol{\mcK_{\mcS\ss\mcX}}\ss\mcK_X$.  By construction we can now delete this subsheaf objectwise on $\Sub(X)$.  That is, we define for $X'\ss X$ $$\mcK_{\mcX\backslash\mcS}(X')=\mcK_X(X')\backslash\mcK_{\mcS\ss\mcX}(X'),$$ where $A\backslash B$ denotes the maximal subset of $A$ which contains no elements in $B$.

The database $$\mcX'\setto (X,\mcK_{\mcX\backslash\mcS},\tau),$$ where $\tau$ is shorthand for $\tau|_{\mcK_{\mcX\backslash\mcS}}\taking\mcK_{\mcX\backslash\mcS}\to\mcU_X$, is the deletion of $\mcS$ from $\mcX$.  There is a canonical map $\mcX'\to\mcX$ in $\Data$, and one can show that $\mcX'$ is the final object under $\mcX$ whose join with $\mcS$ is empty.

\section{Applications, advantages, and further research}\label{sec:applications}

In this section, we discuss the applications of the category of simplicial databases.  First, simplicial databases can be used wherever relational databases are used; though simplicial databases are more general, they are still closed under applying the usual queries.  On the other hand, there are many advantages to using simplicial databases as opposed to relational ones.  

In Section \ref{subsec:geometric}, we discuss how the geometry of a schema can provide an intuitive picture for the content and layout of a database.  As an example of using category theory to reason about databases, we show in Section \ref{subsec:query} that query equivalences are trivially verified when one phrases them in categorical language.  In Section \ref{subsec:privileges} we discuss how diagrams of databases can give various users different privileges in terms of accessing and modifying data.  In Section \ref{subsec:comparison} we address the issue of comparing our categorification of databases to others versions. Finally, in Section \ref{subsec:further research}, we discuss further research on the subject and open questions.

\subsection{Geometric intuition}\label{subsec:geometric}

\hspace{.1in}In Section \ref{subsec:schemas}, we defined the category $\Sch^\pi$ of schemas for a given type specification $\pi$.  They are based on geometric objects called simplicial sets.  In this section, we show that the geometry of these objects is intuitive and therefore useful in practice.

\begin{example}\label{ex:flights}

In this example, we consider a simplified situation in which one keeps track of the cities from which airplane flights take off and those at which they land.  So suppose we have only one type, $\DT=\{\tn{`City'}\}$ and $U$ is the set of cities in the world that have airports.  Let $X$ be the schema $$\xymatrix@1{~^\tn{`City'}\!\bullet\hspace{-.1cm}\ar@{-}@<-.2ex>[r]&\hspace{-.2cm}\bullet^\tn{`City'}}$$  For our sheaf of keys $\mcK$, we take $\mcK(\tn{`City'})=U$.  Over the 1-simplex $X$ take $\mcK(X)$ to be the set of pairs $(c_1,c_2)$ for which $c_1$ is the city of departure and $c_2$ is the city of arrival for some flight.  Let $\mcX$ denote this database of flights.

Now, joining this database with itself yields a database with schema \begin{figure}[htbp]\label{dia:layover}
\setlength{\unitlength}{3947sp}
\includegraphics{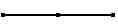}
\put(-2100,50){`City'}
\put(-1200,50){`City'}
\put(-300,50){`City'}
\end{figure}
 whose global sections are ``flights with layover," i.e. pairs of flights with the destination city of the first flight equal to the departing city of the second flight.  Similarly, the database of multi-city trips of a given length $n$ is simply the union (colimit) of $n$ copies of the database of flights $\mcX$ in this way.

Moreover, if we want to use $\mcX$ to find the set of available round-trips, we simply join the ends of the schema in Diagram \ref{dia:layover} to make a circle 

\begin{figure}[htbp]
\setlength{\unitlength}{3947sp}
\includegraphics{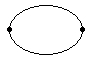}
\put(-1850,450){`City'}
\put(0,450){`City'}
\end{figure}

This is not just heuristic; we have literally taken the indicated limit of databases.  The result is a new database whose global sections are precisely the pairs of flights which constitute a round-trip.

The point is that one can intuit this result by visualizing round-trips as circles, and then applying that vision to the schemas themselves.

\end{example}

\begin{example}\label{ex:sex}

In 2004, Bearman et al. \cite{BMS} present data which shows that at a certain high school called ``Jefferson High," there is a statistically small number of sexual couples that later switch partners.  That is, if $B_1$ and $G_1$ are sexual partners and $B_2$ and $G_2$ are sexual partners, then it rarely happens that later $B_1$ mates with $G_2$ and $B_2$ mates with $G_1$.  As they say ``...we find many cycles of length 4 in the simulated networks, but few in Jefferson..."

Suppose then that we take their raw data and put it on the schema $$\xymatrix@1{~^\tn{`Boyfriend'}\!\bullet\hspace{-.1cm}\ar@{-}@<-.2ex>[r]&\hspace{-.2cm}\bullet^\tn{`Girlfriend'}}$$ which we denote $X$.  Visually, we represent two boys and two girls who switch partners as follows: \begin{eqnarray}\label{dia:switch}\xymatrix{Boys&Girls\\\bullet\hspace{-.1cm}\ar@{-}[r]\ar@{-}[dr]&\hspace{-.1cm}\bullet\\\bullet\hspace{-.1cm}\ar@{-}[ur]\ar@{-}[r]&\hspace{-.1cm}\bullet}\end{eqnarray} (where, say, horizontal lines represent the original partnerships and diagonal lines represent the new partnerships).  And indeed, we can take the union of four copies of $X$ along various vertices to obtain a database with the above 4-cycle schema.  

In other words, there is a way to take raw data over a line segment, representing partnerships, and automatically generate data over the ``switch schema," Diagram (\ref{dia:switch}), just by taking the indicated limit of databases.  The global sections of this new ``switched partners" database are precisely what is being studied in Bearman's paper.

As in Example \ref{ex:flights}, the point is that the shape of the schema is intuitive.  Using schemas that are geometrically intuitive may enhance the ability of users to manipulate and make sense out of the raw data.

\end{example}

\subsection{Query equivalences}\label{subsec:query}

It is well known that joining tables together is very costly.  If one only wishes to consider certain rows or columns of a join, he or she should isolate those rows or columns {\em before} performing the join, not after.  For that reason, one is taught to ``push selects and projects," i.e. to do these operations first.

How does one prove that projecting first and then joining will result in the same database as will joining first and then projecting?  The proofs of results like these are generally tedious.  In this section, we do not claim any new results.  We merely show that these simple query equivalences are obvious when one uses the language of simplicial databases and knows basic category theory. 

For example, it is a standard category-theoretic fact that, in {\em any} category $\mcC$ with limits, there is a natural isomorphism \begin{eqnarray}\label{dia:cat fact}(A\cross_BC)\cross_DE\iso (C\cross_DE)\cross_BA.\end{eqnarray}  Note that both joins and selects are examples of such limits (see Sections \ref{subsec:join} and \ref{subsec:select}).  The formula (\ref{dia:cat fact}) in particular applies to the category $\Data$ of databases and proves that ``selecting $E$ from a join of $A$ and $C$ gives the same result as first selecting $E$ from $C$ and then joining the result with $A$.

Projecting a database to a subschema is easy to describe in the theory of simplicial databases: one simply restricts the sheaf $\mcK$ and the map $\tau$ to that subschema (see Section \ref{subsec:projections}).  The fact that projects commute with joins follows from basic sheaf theory, e.g. that the limit of a diagram of sheaves is the same as the limit of the underlying diagram of presheaves. 

\subsection{Privileges}\label{subsec:privileges}

The sheaf-theoretic nature of our conception of databases lends itself nicely to the idea of privileges.  It often happens that one wishes to give a particular user the ability to modify certain sections of the database but not others.  If $X$ is the schema for a database $\mcX$, perhaps we wish to give a particular user the ability to modify data on the subschema $i\taking X'\ss X$.  

To accomplish this, note that there is a map of databases $$\mcX=(X,\mcK_X,\tau_X)\too(X',i^*\mcK_X,i^*\tau_X)=\mcX'$$  We allow the user to see $\mcX'$ as a database and make changes to it (we could also limit the ways in which this user can modify $\mcX'$ -- only allow insertions, for example).  At any given time, the user only sees the sub-database $\mcX'$.  

Suppose he or she adds a few lines to the sheaf $i^*\mcK_X$ to make it $i^*\mcK_X\cup\mcL$.  To update the main database, we take the colimit of the diagram of sheaves $$\xymatrix{i_!i^*\mcK_X\ar[r]\ar[d]&i_!(i^*\mcK_X\cup\mcL)\\\mcK_X}$$ and the result will be a new sheaf on $X$ with the appropriate insertions.  

Deletions are handled in a somewhat different way, but the idea is the same.  If the user deletes data from the sheaf $i^*\mcK_X$ to obtain the sheaf $i^*\mcK_X\backslash\ol{\mcD}$, then to update the main database may require us to delete entries from larger schemas (see Section \ref{subsec:delete}).  The updated sheaf on $X$ will be the limit of the diagram $$\xymatrix{&\mcK_X\ar[d]\\i_+(i^*\mcK_X\backslash\ol{\mcD})\ar[r]&i_+i^*\mcK_X.}$$

Again, we are not claiming that privileges of this type are anything new.  We are claiming that they are naturally phrased in this categorical language, thus bringing a new and powerful mathematical tool to bear on the problems of the subject.

\subsection{Comparison to other categories of databases}\label{subsec:comparison}

As mentioned in the introduction, many other categorifications of databases have been presented over the years.  One of the nice features of category theory is that one can compare various categories using functors.  Given another categorical formulation of databases, we could try to produce a functor from it to $\Data$ and from $\Data$ back to it.  The way that these functors behave (e.g. if they are adjoint, or if one or the other is fully faithful) will tell us about the relative expressive power of the models, as well as understand how to translate between them.  We hope to work on such a comparison in the future.

\subsection{Further research}\label{subsec:further research}

The category-theoretic and also geometric nature of simplicial databases opens up many directions for future research.  We present a few in this subsection that we intend to pursue.  Many of these ideas were suggested to us by Paea LePendu.

\subsubsection{Topological methods}\label{subsubsec:top}

First, we would like to consider how we might use methods from algebraic topology to study databases.  Recall from Example \ref{ex:simplices} that there is a functor $\Sch\to\Top$ called {\em topological realization} that allows one to naturally view any schema as a topological space.  Furthermore, we already saw in Example \ref{ex:sex} that importing topological ideas can have real world meaning: topological 4-cycles represented pairs of mating couples that switched partners.  

Another example of the usefulness of topological methods is given by ``lifting problems."  Problems of this sort include the famous question ``are there three foods, each pair of which taste good when eaten together, but the threesome of which tastes bad when eaten together?"  

To phrase this in terms of social networks, suppose that for any $n$ people, either this group is said to be a friendship group or it is not.  The above lifting problem becomes: ``are there three people, each pair of which is a friendship group, but the triple is not?"  These types of phenomena can be represented geometrically, so having simplicial sets as schema may be useful for their study.

Homotopical methods from algebraic topology may also be useful.  When one object ``morphs" into another over the course of time (such as a child becoming an adult), it is difficult to know how to treat that object in a database.  Homotopy theory is the study of gradual transformation through time, and the author sees some potential for using it to study real-world phenomena.

Finally, the geometric nature of our schema may be useful for query optimization.  Schemas can be classified according to their geometric structure.  It may be that in performing many queries, a database management system learns that some geometric structures are being used more often than others.  The patterns which emerge may be only visible when one uses schemas that have this higher dimensional geometric nature.

\subsubsection{Functional dependencies and normal forms}

In this paper we have not discussed functional dependencies or normal forms.  It is appealing to ask the following question:

\begin{question}

Let $X\in\Sch$ denote a schema; it should be thought of as having a shape (again, via the topological realization functor $\Sch\to\Top$), namely a union of tetrahedra.  We wonder:

\begin{enumerate}\item Given a set of functional dependencies, is there a natural way to annotate the shape $X$ so that these dependencies are made visual? \item Given a schema $X$ that has been annotated in this way, can one easily determine whether it is in a certain normal form?  \item If an annotated schema is not in normal form, do the annotations help in finding the normalization? \end{enumerate}  If the answer to these questions is affirmative, we will have more evidence that the geometric nature of our schema is useful for database design and management.

\end{question}

We hope to address these questions in the near future.

\subsubsection{Database integration}\label{subsubsec:database integration}

We believe that having a rigorous definition for {\em morphisms of databases} (see Definition \ref{def:database morphisms}) will be of use in the problem of database integration.  The morphisms of databases can account for simultaneous changes in schema and in data.  It is also easy to allow changes in data types as well, a topic we will address in later work.

Also, as mentioned in Remark \ref{rem:internal keys} and Section \ref{subsec:duplication}, the use of internal keys should prove immensely valuable.  Instead of including an arbitrarily chosen identifier for an object as part of the data for that object, as required in the theory of relational databases, our theory keeps these arbitrary identifiers separate.  When attempting to integrate databases, it is imperative that one know which sections of the data are {\em observed and invariant properties} of the objects being classified, and which sections of the data are {\em arbitrarily assigned} for management reasons.  Our theory keeps these sections of the data distinct, by use of a sheaf of keys $\mcK$ that is not considered part of the data.

In future research, we hope to show that database integration is made substantially easier when one works with a rigorous and geometric model like the one we present here.  Before we do so, we need to explain how to work with a change in type specifications, which is not hard, and how to deal with constraints in the data.  See Section \ref{subsubsec:types} for our plans in this direction.

\subsubsection{Ontologies and networks}

One intuitively knows that there is a connection between databases and ontologies.  An ontology is meant for organizing knowledge, a database is meant for organizing information, and there is a strong correlation between the two.  In order to make this correlation precise, one must first find precise definitions of ontologies and databases.  Further, these definitions should be phrased in the same language so that they can be compared.  Category theory was invented for the purposes of comparing different mathematical structures, and as such provides a good setting for this project.

Our plan (see \cite{Spi2})) for a categorical definition of communication networks involves annotating the simplices of a simplicial set with databases.  That is, each node in a network has access to a database of ``what it knows," and connections between nodes allows communication via a common language and set of shared knowledge.  In order to make this precise, we need a precise definition for a category of databases, for which Definition \ref{def:database morphisms} suffices.  

\subsubsection{More exotic types}\label{subsubsec:types}

Throughout this paper, we have fixed a type specification $\pi\taking U\to\DT$, where $\DT$ is a set of data types, and $U$ is the disjoint union of the corresponding domains.  This allows for types like strings, characters, dates, integers, etc.  It also allows for more general types like ``functions from $A$ to $B$" or ``probability distributions on a space."  

However, as flexible as our type specifications may be, the situation can be generalized considerably by allowing $\pi$ to be a functor between categories, rather than a function between sets.  The simplest application is one that is already implicitly used, namely sorting data.  The set of strings is in fact an ordered set, and so can be represented as a category (with a morphism from $A$ to $B$ if $B$ is lexicographically larger than $A$).  Another application comes from putting constraints in the data, like if we only allow (city, state) pairs for which the city is within the state. 

By generalizing type specifications to include categories rather than sets, we open up many new possibilities for making sense of data.  Causal relationships can be represented, as can processes.  In short, morphisms make the theory more dynamic.

\bibliographystyle{amsalpha}
\bibliography{biblio}

\end{document}